\documentclass[sigconf]{acmart}
\usepackage[utf8]{inputenc}

\usepackage{mathdots}
\usepackage{mathrsfs}

\usepackage{hyperref}
\hypersetup{
	colorlinks=true,     
	citecolor=blue,      % couleur des liens de citations
	linkcolor=black,      % couleur des liens internes
}

\usepackage{graphicx}
\usepackage{tikz}
\usepackage[ruled,linesnumbered]{algorithm2e}

\SetKwInOut{Input}{Input}
\SetKwInOut{Output}{Output}
\SetKwInOut{Oracle}{Oracle}

\usepackage{enumitem}

\newtheorem{problem}{Problem}
\newtheorem{theorem}{Theorem}[section]
\newtheorem{proposition}[theorem]{Proposition}
\newtheorem{definition}[theorem]{Definition}
\newtheorem{lemma}[theorem]{Lemma}
\newtheorem{corollary}[theorem]{Corollary}

\newtheorem{assumption}{Assumption}%[section]
\makeatletter
\DeclareRobustCommand\onedot{\futurelet\@let@token\@onedot}
\def\@onedot{\ifx\@let@token.\else.\null\fi\xspace}
\def\ie{\emph{i.e}\onedot} 
\makeatother

\renewcommand{\bar}{\overline}

\newcommand{\mbf}{\boldsymbol}
\newcommand{\dd}{d}
\renewcommand{\k}{\kappa}
\newcommand{\BC}{\mathbb{C}}
\newcommand{\BR}{\mathbb{R}}
\newcommand{\BQ}{\mathbb{Q}}
\newcommand{\BK}{\mathbb{K}}
\newcommand{\AK}{\mathcal{A}_\BK}
\newcommand{\softO}{\widetilde{\mathcal{O}}}
\DeclareMathOperator{\TQ}{TaQ}
\DeclareMathOperator{\sign}{sign}
\DeclareMathOperator{\herm}{herm}
\DeclareMathOperator{\Herm}{Herm}
\DeclareMathOperator{\tr}{Tr}
\DeclareMathOperator{\rk}{rk}
\DeclareMathOperator{\Signature}{Sign}
\DeclareMathOperator{\Var}{Var}
\DeclareMathOperator{\GL}{GL}

\DeclareMathOperator{\lm}{lm}
\DeclareMathOperator{\lc}{lc}
\DeclareMathOperator{\Reali}{Reali}
\DeclareMathOperator{\SIGN}{SIGN}
\DeclareMathOperator{\Mat}{Mat}
\DeclareMathOperator{\Ada}{Ada}
\newcommand{\lex}{\mathrm{lex}}
\newcommand{\GB}{\mathscr{G}}
\newcommand{\Basis}{\mathscr{B}}

\author{Louis Gaillard}
\affiliation{%
  \institution{ENS de Lyon, %\textsc{CNRS}
  }
        \city{F-69007 Lyon}
        \country{France}}
\author{Mohab Safey El Din}
\affiliation{%
        \institution{Sorbonne Universit\'e, \textsc{CNRS}, \textsc{LIP6}}
        \city{F-75005 Paris}
        \country{France}}
\renewcommand{\shortauthors}{L. Gaillard, and M. Safey El Din}
\title[Solving parametric semi-algebraic systems]{Solving parameter-dependent semi-algebraic systems}

\copyrightyear{2024}
\acmYear{2024}
\setcopyright{licensedothergov}\acmConference[ISSAC '24]{International Symposium on Symbolic and Algebraic Computation}{July 16--19, 2024}{Raleigh, NC, USA}
\acmBooktitle{International Symposium on Symbolic and Algebraic Computation (ISSAC '24), July 16--19, 2024, Raleigh, NC, USA}
\acmDOI{10.1145/3666000.3669718}
\acmISBN{979-8-4007-0696-7/24/07}

\begin{abstract}
  We consider systems of polynomial equations and inequalities in $\BQ[
  \mbf{y}][\mbf{x}]$ where $\mbf{x} = (x_1, \ldots, x_n)$ and $\mbf{y} = (y_1, \ldots,
  y_t)$. The $\mbf{y}$ indeterminates are considered as \emph{parameters} and we
  assume that when specialising them \emph{generically}, the set of common
  complex solutions, to the obtained equations, is finite.

  We consider the problem of real root classification for such
  parameter-dependent problems, i.e. identifying the possible number of real
  solutions depending on the values of the parameters and computing a
  description of the regions of the space of parameters over which the number of
  real roots remains invariant.

  We design an algorithm for solving this problem. The formulas it outputs enjoy
  a determinantal structure. Under genericity assumptions, we show that its
  arithmetic complexity is polynomial in both the maximum degree $d$ and the
  number $s$ of the input inequalities and exponential in $nt+t^2$. The output
  formulas consist of polynomials of degree bounded by $(2s+n)d^{n+1}$. This is
  the first algorithm with such a singly exponential complexity. We report on
  practical experiments showing that a first implementation of this algorithm
  can tackle examples which were previously out of reach.
\end{abstract}

\begin{document}

% <concept_id>10010147.10010148.10010149.10010150</concept_id>
% <concept_desc>Computing methodologies~Algebraic algorithms</concept_desc>
% <concept_id>10003752.10003809</concept_id>
% <concept_desc>Theory of computation~Design and analysis of algorithms</concept_desc>
%\end{CCSXML}


% :>> acmart stuff
%\CopyrightYear{2024}
%\conferenceinfo{ISSAC '24,}{July 16--19, 2024, Raleigh, NC, USA}
%\isbn{979-8-4007-0696-7/24/07}
%\doi{https://doi.org/10.1145/3666000.3669718}

\maketitle

\section{Introduction}

\paragraph*{Problem statement}
We consider polynomials $\mbf{f} = (f_1, \dots, f_e)$ and $\mbf{g} = (g_1,
\dots, g_s)$ in $\BQ[\mbf{y}][\mbf{x}]$ with $\mbf{x} = (x_1, \dots, x_n)$ and
$\mbf{y} = (y_1, \dots, y_t)$. The variables $\mbf{x}$ (resp. $\mbf{y}$) are
seen as the \emph{unknowns} (resp. \emph{parameters}) of the system. Further, we
denote by $\pi$ the canonical projection $(\mbf{y}, \mbf{x}) \to \mbf{y}$ on the
space of parameters. We denote by $\mathcal{V} \subseteq \BC^{t + n}$ the
(complex) algebraic variety defined by $\mbf{f} = 0$, and by $\mathcal{V}_\BR$
its real trace $\mathcal{V}\cap \BR^{t + n}$.
In this paper, we assume the following.
\begin{assumption}
  \label{assump:finite_nb_sol_generically}
  There exists a nonempty Zariski open subset $\mathcal{O} \subseteq \BC^t$
  such that for all $\eta \in \mathcal{O}$,
  $\pi^{-1}(\eta) \cap \mathcal{V}$ is nonempty and finite. 
\end{assumption}
In other words, for a \emph{generic} specialization point $\eta$, the
specialized system $\mbf{f}(\eta, \cdot) = 0$ is zero-dimensional. Besides, can
assume that the cardinality of $\pi^{-1}(\eta) \cap \mathcal{V}$ remains
invariant when $\eta$ ranges over $\mathcal{O}$. This is not the case for the
set of real solutions.

We consider the (basic) semi-algebraic set $\mathcal{S} \subseteq \BR^{t +n}$
defined by
\begin{equation}
  \label{eq:sa_system_fg}
  f_1 = \dots = f_e = 0, \quad g_1 > 0, \dots,~ g_s > 0.
\end{equation}
The goal of this paper is to provide an efficient algorithm for solving the real
root classification problem over $\mathcal{S}$ as stated below.
\begin{problem}[Real solution classification]
  \label{pb:solution_classification}
  On input $(\mbf{f}, \mbf{g})$ with $\mbf{f}$
  satisfying Assumption~\ref{assump:finite_nb_sol_generically}, compute $(\Phi_i, \eta_i,
  r_i )_{1 \le i \le \ell }$ where for all $1\le i \le \ell$, $\Phi_i$ is a
  semi-algebraic formula in $\BQ[\mbf{y}]$ defining the semi-algebraic set
  $\mathcal{T}_i \subseteq \BR^t$, with $\eta_i \in \mathcal{T}_i$ and $r_i \ge
  0$ such that
  \begin{itemize}
  \item for all $ \eta \in \mathcal{T}_i$,
    the number of points in $\mathcal{S} \cap \pi^{-1}(\eta)$ is $r_i$,
  \item $\bigcup_{i=1}^\ell \mathcal{T}_i$ is dense in $\BR^t$.
  \end{itemize}
\end{problem}
Such a sequence $(\Phi_i, \eta_i, r_i )_{1 \le i \le \ell }$ is said to be a
solution to Problem~\ref{pb:solution_classification} which arises in many applications
(see e.g. \cite{BFJSV16, puente2023absolute, Co02, FMRSa08, Yang00,
  le:hal-03283378}).
\paragraph*{Prior works} First, as noticed in
\cite{Le2022RealRootClassification}, the cylindrical algebraic decomposition
(CAD) algorithm due to Collins \cite{collins1975cad} could be used to solve
Problem~\ref{pb:solution_classification}. However, its doubly exponential
complexity \cite{dav07, Dav88} in the total number of variables makes it
difficult to use.

More efficient approaches have been devised by using polynomial elimination
methods combined with real algebraic geometry. They consist in computing some
nonzero polynomials, say $h_1, \ldots, h_k$ in $\BQ[\mbf{y}]$, such that the
number of points in $\mathcal{S}\cap \pi^{-1}(\eta)$ remains invariant when
$\eta$ ranges over some connected component of the semi-algebraic set defined by
$h_1\neq 0, \ldots, h_k\neq 0$. Such polynomials are called \emph{border
  polynomials}, in the context of methods using the theory of regular chains
(see e.g. \cite{YangHX01, YangXia05, Li08, Chen13}), or \emph{discriminant polynomials}
in the context of methods using algebraic elimination algorithms based on
Gr\"obner bases (see e.g. \cite{LaRo07, Moroz06}) when the ideal generated by
$\mbf{f}$ is assumed to be radical and equidimensional. When $\dd$ is the
maximum degree of the input polynomials in $\mbf{f}$ and $\mbf{g}$, these
$h_i$'s can be proven to have degree bounded by $n(\dd-1)\dd^{n}$.

Once these polynomials are computed one then needs to describe the connected
components of the set where none of them vanish. When this is done through the
CAD algorithm, the cost of this is doubly exponential in $t$, the number of
parameters. Using more advanced algorithms for computing semi-algebraic
descriptions of connected components of semi-algebraic sets (see \cite[Chap.
16]{roy2006algo_real_alg_geom}) through parametric roadmaps,
one can obtain a complexity using $\left( n(\dd-1)\dd^{n} \right)^{O(t^{4})}$
arithmetic operations in $\BQ$ and which would output polynomials of degree
lying in $\left( n(\dd-1)\dd^{n} \right)^{O(t^{3})}$.

All in all, only little is known about the complexity of these methods and it has
been an open problem to obtain better complexity estimates or degree bounds on
the polynomials of the output formulas required to solve
Problem~\ref{pb:solution_classification}.

A first step towards this goal comes from the analysis of the algorithm in
\cite{Le2022RealRootClassification}. This algorithm is restricted to the case
where the ideal generated by $\mbf{f}$ is radical and the sequence $\mbf{g}$ is
empty. Under \emph{genericity assumptions} on the input $\mbf{f}$, this
algorithm runs in time quasi-linear in $n^{O(t)}\dd^{3nt+O(n+t)}$ and the
degrees of the polynomials in the output formulas lie in $n(\dd-1)\dd^{n}$. This
is achieved using classical real root counting methods (through Hermite's
quadratic forms) but combined in an innovative way with the theory of Gr\"obner
bases. Additionally, the output formulas enjoy a nice determinantal encoding
which allows one to evaluate them easily. This is at the foundations of new
efficient algorithms for one-block quantifier elimination
\cite{Le2021QuantifierElim}. We also note that these techniques can lead to a new
geometric approach for Cylindrical Algebraic Decomposition
\cite{chen2023geometric}.

Still, several open problems remain. One is to obtain similar complexity bounds
which do not depend on the aforementioned genericity assumptions. Another open
problem is to extend such an approach to real root classification problems
\emph{involving inequalities}, hence extending significantly the range of
applications which could be reached. In this paper, we tackle this second open
problem.

% In \cite[Sec.~3]{Le2022RealRootClassification}, an algorithm is described for
% solving the following problem. Given $(h_1, \dots, h_\ell ) \in \BQ[y_1, \dots,
% y_t]$, compute at least one point per connected component of the semi-algebraic
% set $\mathcal{S} \subset \BR^t$ defined by
% \begin{equation*}
%   h_1 \neq 0, \dots, h_\ell \neq 0.
% \end{equation*}
% %
% The main ideas for this algorithm follow from~\cite{SafeySchost2003OnePointPerConnectedComponent}.
% We will rely on the following result.
% \begin{theorem}[{\cite[Thm.~2]{Le2022RealRootClassification}}]
%   \label{thm:sample_points_connected_component}
%   Let $(h_1, \dots, h_\ell) \subset \BQ[y_1, \dots, y_t]$ with $D$ an
%   upper bound for the degree of the $h_i$'s.  Consider the
%   semi-algebraic set $\mathfrak{S}$ defined by the non-vanishing set of
%   $h_1, \dots, h_\ell$.  There exists a probabilistic algorithm which
%   on input $(h_1 ,\dots, h_\ell)$ outputs a representation of a finite
%   subset $\mathfrak{R}$ of $\BQ^t$ with at most $(2\ell D)^t$ elements
%   such that $\mathfrak{R}$ meets every connected component of
%   $\mathfrak{S}$ using
%   $\widetilde{\mathcal{O}}\left( \binom{D+t}{t} \ell^{t+1}2^{3t}D^{2t
%       +1}\right)$ arithmetic operations in $\BQ$.
% \end{theorem}
 
\paragraph*{Contributions}
We present an algorithm solving Problem~\ref{pb:solution_classification} revisiting the
ideas in~\cite{Le2022RealRootClassification} to handle the case of systems of
equations and inequalities. It uses a real root counting machinery based on
Hermite's quadratic form~\cite{hermite1856realrootcounting} in some appropriate
basis. In order to take the polynomial inequalities defined by $\mbf{g}$ into
account, this algorithm relies on a method originated
in~\cite{BenOrKozenReif1986comp_elem_algebra} using the so-called
Tarski-query~\cite[Sec.~10.3]{roy2006algo_real_alg_geom} for determining the
sign conditions realized by a family of polynomials on a finite set of points.
These methods are devised to count the number of real solutions to some system
of polynomial equations (with coefficients in $\BR$), with finitely many complex
roots, which do satisfy some extra polynomial inequalities.

Our contribution combines these methods with Gr{\"o}bner bases computations in our
context where the coefficients of our input polynomials depend on the parameters
$\mbf{y}$. A second key ingredient, used to control the number of calls to
Tarski queries, in a way that is similar to the one used in \cite[Chap.
10]{roy2006algo_real_alg_geom}, is the use of efficient routines for computing
sample points per connected components in semi-algebraic sets lying in the space
of parameters \cite[Sec. 3]{Le2022RealRootClassification}. Hence, the
semi-algebraic constraints depending on $\mbf{y}$ actually encode some constraints
on the signature of \emph{parameter-dependent} Hermite matrices and thus, enjoy
a nice determinantal structure.

Note that, by contrast with \cite{LaRo07}, this algorithm does not assume that
the ideal generated by $\mbf{f}$ is radical and equidimensional.

Since this algorithm makes use of Gr{\"o}bner bases computations, in order to
control the complexity, extra genericity
assumptions are needed. Hence, for the purpose of a complexity analysis, we assume that the
homogeneous component of the $f_i$'s of highest degree forms a regular sequence,
which we abbreviate by the saying that $\mbf{f}$ is a regular sequence. In
addition, letting $\GB$ be a reduced Gr{\"o}bner basis for the ideal generated by
$\mbf{f}$ and the graded reverse lexicographical ordering, we assume the
following.

\begin{assumption}
  \label{assump:generic_degree_grobner_basis}
  For any $p \in \GB$, we have $\deg p = \deg_{\mbf{x}} p $.
\end{assumption}

These assumptions are known to be satisfied generically (see \cite[Prop.
20]{Le2022RealRootClassification}) and to enable nicer complexity bounds on
Gr{\"o}bner bases. Our main complexity result is the following one. We use the notation $g = \softO(f)$
meaning that $g = \mathcal{O}(f \log^\k(f))$ for some $\k > 0$.

\begin{theorem}\label{thm:main}
  Let $\mbf{f} \subset \BQ[\mbf{y}][\mbf{x}]$ be a regular sequence such that
  $\mbf{f}$ satisfies both
  Assumptions~\ref{assump:finite_nb_sol_generically} and~\ref{assump:generic_degree_grobner_basis}.
  Let  $\mathfrak{D} \coloneqq (2sd + n(d-1))d^{n}$.
  %\lambda)d^n$ and $r \coloneqq \binom{s}{t} 4^{t+1} d(2d -1)^{n+t -1}$.
  % Then \textcolor{magenta}{in case of success},
  % \cref{alg:classification}
  There exists an algorithm which computes a solution to
  Problem~\ref{pb:solution_classification} using
  \begin{equation*}
    % \textcolor{red}{
    \softO\left(  \binom{t + \mathfrak{D}}{t} \binom{s}{t}^{t+1}
      2^{3t^2 + nt + 8t + n} s^{t+2} (2s + n)^{2t+1}d^{t^2 + 4nt + 3(n+t) + 1}\right)
  %}
  \end{equation*}
  arithmetic operations in $\BQ$ and outputs at most %$(4d^ns \rho \mathfrak{D})^t$
  $(sdn)^{\mathcal O(t^2 + nt)}$
  formulas that consists of $(s d)^{\mathcal{O}(n+t)}$ % $sd^n \rho$
  polynomials of degree at most
  $\mathfrak{D}$.
\end{theorem}

\noindent Because the binomial coefficient $\binom{t+ \mathfrak D}{t}$
can be bounded by $2 \mathfrak D^t$ and
$\binom{s}{t}$ by $s^t$, the complexity estimate in Theorem~\ref{thm:main}
is bounded by $\softO(2^{3t^2 + nt}s^{t^2}n^{3t}d^{t^2 + 5nt})$.

We report on practical experiments performed with a first implementation of this
algorithm. That implementation makes no use of the genericity Assumption~\ref{assump:generic_degree_grobner_basis}
made
to enable a complexity analysis, it only uses
Assumption~\ref{assump:finite_nb_sol_generically}.
Practical performances which are achieved show that this algorithm
outperforms the state-of-the-art software for solving
Problem~\ref{pb:solution_classification}. We also fully solve an instance
of the Perspective-Three-Point problem for which computing
semi-algebraic formulas for the real root classification was an open problem.

% \todo[inline]{Pas forc\'ement ici mais quelque part dans le papier pointer le
%   fait qu'on relaxe les hypoth\`eses des algorithmes pr\'ec\'edents (plus besoin
%   de supposer que les intersections sont transverses)}

% \todo[inline]{Quelque part dans le papier, rendre limpide le fait que nos
%   hypoth\`eses sont satisfaites g\'en\'eriquement}

\paragraph*{Plan of the paper} Section~\ref{sec:hermite} recalls the basics on
Hermite's quadratic forms, using materials mostly from
\cite[Chap.~4]{roy2006algo_real_alg_geom}. Section~\ref{sec:paramHermite}
generalizes constructions and results on \emph{parametric} Hermite matrices from
\cite{Le2022RealRootClassification} to the case where inequalities are involved.
Section~\ref{sec:algo} describes the algorithm on which Theorem~\ref{thm:main}
relies and proves the complexity statements. Section~\ref{sec:experiments}
reports on practical experiments.

\section{Hermite's quadratic form}\label{sec:hermite}

We recall some basic definitions and properties on Hermite's quadratic forms. For
more details, we refer to~\cite[Chap.~4]{roy2006algo_real_alg_geom}.

\subsection{Definition}

Let $\BK$ be a field of characteristic 0 and $\mbf{f} = (f_1, \dots, f_e)\subset
\BK[\mbf{x}]$ be generating a zero-dimensional ideal of degree $\delta$, denoted by
$\left<\mbf{f}\right>_\BK$. The quotient ring $\AK \coloneqq \BK[\mbf{x}]/\left<
  \mbf{f}\right>_\BK$ is a finite dimensional $\BK$-vector space \cite[Theorem
4.85]{roy2006algo_real_alg_geom} of dimension $\delta$.

A monomial basis $B =
(b_1, \dots, b_\delta)$ of $\AK$ can be derived from a Gr{\"o}bner basis $G$ of
$\left< \mbf{f}\right>_\BK$ with respect to (w.r.t.) an admissible monomial
ordering $\succ$ over $\BK[\mbf{x}]$: it is the set of monomials that are not
divisible by any leading monomial of elements in $G$. For $q \in \BK[\mbf{x}]$,
we denote by $\bar{q}$ the class of $q$ in $\AK$ and by $L_q\colon \bar{p} \in
\AK \mapsto \bar{p \cdot q} \in \AK$ the multiplication by $q$ in $\AK$.

\begin{definition} [Hermite's quadratic form]
  \label{def:hermite_quadratic_form}
  For $g \in \AK$, \emph{Hermite's bilinear form} $\herm(\mbf{f},g)$ is defined
  by
  \begin{align*}
    \herm(\mbf{f},g) \colon \AK \times \AK &\rightarrow \BK \\
    (p,q) &\mapsto \tr(L_{g p q}),
  \end{align*}
  where $\tr$ denotes the trace. The corresponding quadratic form $p \mapsto
  \tr(L_{gp^2})$ is called \emph{Hermite's quadratic form} $\Herm(\mbf{f},g)$.
  The \emph{Hermite matrix} associated to $(\mbf{f},g)$ w.r.t. the basis $B$ is
  the matrix $\mathcal{H} = (\tr(L_{g
    b_i b_j}))_{1 \le i,j \le \delta} \in \BK^{\delta
    \times \delta}$ of $\Herm(\mbf{f},g)$ in the basis $B$.
\end{definition}

% For $g \in A$, one relates the Hermite matrices associated to $\Herm(\mbf{f},g)$
% and $\Herm(\mbf{f},1)$ thanks to the matrix of the multiplication $L_g$.
\noindent
For a matrix $\mathcal{H}$, $\mathcal{H}_{i,j}$ is its element at the $i$-th row
and $j$-th column.

\begin{proposition}
  \label{prop:relation_H1_HQ}
  Let $B = (b_1, \dots, b_\delta)$ be a basis of $\AK$, $g \in \AK$,
  $\mathcal{H}_1$ and $\mathcal{H}_g$ be the matrices of $\Herm(\mbf{f},1)$ and
  $\Herm(\mbf{f},g)$ w.r.t. $B$. Let $M =(m_{i,j})$ be the matrix of $L_g$
  w.r.t. $B$. Then,
  %\begin{align*}
    $\mathcal{H}_g = \mathcal{H}_1 M$.
  %\end{align*}
\end{proposition}

\begin{proof}
  For $p,q \in \AK$, we have $\herm(\mbf{f},g)(p,q) = \tr(L_{gpq}) =
  \herm(\mbf{f},1)(p,gq)$. Thus, it holds that
  \begin{align*}
    (\mathcal{H}_g)_{i,j} &= \herm(\mbf{f},1)(b_i,gb_j) =
                            \herm(\mbf{f},1)\left( b_i, \sum_{k=1}^\delta m_{k,j}b_k\right)\\
                          &  = \sum_{k=1}^\delta m_{k,j}(\mathcal{H}_1)_{i,k} = (\mathcal{H}_1 M)_{i,j}. \qedhere    
  \end{align*}
\end{proof}

\subsection{Real root counting}

For now, we assume $\BK = \BR$ or $\BQ$ (or any ordered field).
For $z \in \BK$, $\sign(z)$ is $-1,0$ or $1$ if $z< 0$, $z= 0$ or $z > 0$
respectively.

\begin{definition}[Tarski-query]
  \label{def:tarski_query}
  Let $Z$ be a finite set in $\BK^n$ and $g \in \BK[\mbf{x}]$.
  %Let $\mathcal{V}_\BR$ be the finite set of real solutions of the system.
  We define the \emph{Tarksi-query}, $\TQ(g,Z)$ of $g$ for $Z$ by
  \begin{align*}
     \sum_{x \in Z} \sign(g(x))  =
    \sharp \{ x \in Z \mid g(x) > 0 \} 
    - \sharp  \{ x \in Z \mid g(x) <0 \}.
  \end{align*}
  When $Z$ is the finite set of real roots of a zero-dimensional system $\mbf{f}
  = 0$, we denote it by $\TQ(g,\mbf{f})$.
\end{definition}

We denote the signature of a real quadratic form $q$ by $\Signature(q)$.
\begin{theorem} [{\cite[Thm.~4.100]{roy2006algo_real_alg_geom}}]
  \label{thm:signature_tarski_query}
  Let $\mbf{f} = (f_1, \dots, f_e) \subset \BK[\mbf{x}]$ be as above and $g \in
  \BK[\mbf{x}]$. Then,
  %\begin{align*}
    $\Signature(\Herm(\mbf{f},g)) = \TQ(g, \mbf{f})$.
  %\end{align*}
\end{theorem}

Hence, Tarski-queries are given by signatures of Hermite matrices. From
$\TQ(1,\mbf{f}) \TQ(g,\mbf{f})$ and $\TQ(g^2,\mbf{f})$, one can compute the
number of real roots of $\mbf{f}$ that satisfy a given sign condition for $g$.
We define $c(g ~\diamondsuit~ 0)$ for $\diamondsuit \in \{ <, =, >\}$ as $ \sharp
\{x \mid \mbf{f}(x) = 0 \land g(x) ~\diamondsuit~ 0 \}$. We have the following
invertible system
\begin{align}
  \label{eq:sign_determination_system_1_poly}
  \begin{bmatrix} 1 & 1  & 1 \\ 0 & 1  & - 1 \\ 0 & 1 & 1 \end{bmatrix} \cdot
  \begin{bmatrix}  c(g = 0) \\ c(g > 0 ) \\ c(g < 0) \end{bmatrix} = \begin{bmatrix}
    \TQ(1,\mbf{f}) \\ \TQ(g,\mbf{f}) \\ \TQ(g^2,\mbf{f})
  \end{bmatrix}.
\end{align}

% For a square matrix $M$, the $i$-th leading principal minor $M_i$ of
% $M$ is the determinant of the submatrix obtained from the first $i$ rows and
% columns of $M$. By convention, $M_0 =1$.

\begin{theorem}[{\cite[Thm.~1.9]{ghys2016signatures}}]
  \label{thm:signature_symmetric_matrix}
  Let $M$ be a symmetric matrix in $\BR^{m \times m}$ and for $0 \le i \le m$ let $M_i$ denote
  its $i$-th leading principal minor.
  We assume that $M_i \neq 0$ for all $0 \le i \le m$. Then we have
  \begin{equation*}
    \Signature(M) = m - 2 \Var(M_0, M_1, \dots, M_m),
  \end{equation*}
  where $\Var$ stands for the number of sign variations in the sequence.
\end{theorem}

\section{Parametric Hermite matrices}\label{sec:paramHermite}

\subsection{Basic construction and properties}\label{ssec:basicparam}

Let $\mbf{f} = (f_1, \dots, f_e) \subset \BQ[\mbf{y}][\mbf{x}]$ such that $\mbf{f}$
satisfies Assumption~\ref{assump:finite_nb_sol_generically} and $g \in  \BQ[\mbf{y}][\mbf{x}]$.
%One can think that $g$ is one of the $g_i$'s in an instance of \cref{pb:solution_classification}.
We take as a base field $\BK$ the rational function field $\BQ(\mbf{y}) $.
By~\cite[Lem.~4]{Le2022RealRootClassification}, as $\mbf{f}$ satisfies
Assumption~\ref{assump:finite_nb_sol_generically}, the ideal $\left< \mbf{f} \right>_\BK$
generated by $\mbf{f}$ in $\BK[\mbf{x}]$ is zero-dimensional. Hence one can
define Hermite's quadratic form $\Herm(\mbf{f},g)$ and compute a
\emph{parametric} Hermite matrix $\mathcal{H}_g \in \BK^{\delta \times \delta}$
representing $\Herm(\mbf{f},g)$ (once a basis $B$ of $\AK$ is fixed). We start
by making explicit how these can be computed and next prove some \emph{nice}
specialization properties of these matrices. We follow and extend the approach
in~\cite{Le2022RealRootClassification}.

\paragraph*{Gr\"obner basis and monomial basis}
We denote $\texttt{grevlex}(\mbf{x})$ for the graded reverse lexicographical
ordering (\emph{grevlex}) among the variables $\mbf{x}$ (with $x_1 \succ \dots
\succ x_n$) and $\texttt{grevlex}(\mbf{x}) \succ \texttt{grevlex}(\mbf{y})$
(with $y_1 \succ \dots \succ y_t$) for the elimination ordering.
For
$p \in \BC(\mbf{y})[\mbf{x}]$, $\lc_{\mbf{x}}(p)$ (resp. $\lm_{\mbf{x}}(p)$) denotes the leading
coefficient (resp. monomial) of $p$ for the ordering  $\texttt{grevlex}(\mbf{x})$.
%and $\lm_{\mbf{x}}(p)$ denotes its leading monomial.
We let $\GB$ be
the reduced Gr\"obner basis of $\langle \mbf{f} \rangle\subset \BQ[\mbf{y},
\mbf{x}]$ w.r.t. this elimination ordering. By
\cite[Lem.~6]{Le2022RealRootClassification}, $\GB$ is also a Gr{\"o}bner basis
of $\left< \mbf{f} \right>_\BK $ w.r.t. $\texttt{grevlex}(\mbf{x})$.
Hence the set $\Basis$ of all monomials in $\mbf{x}$ that are not reducible by
the leading monomials of $\GB$ (w.r.t. $\texttt{grevlex}(\mbf{x})$) is finite
since $\left< \mbf{f} \right>_\BK $ is zero-dimensional. It forms a basis of
$\AK$. We define $\mathcal{H}_g$ as the parametric Hermite matrix associated to
$(\mbf{f},g)$ w.r.t. the basis $\Basis$.

% The basis $B$ of $\AK$ is obtained from the reduced Gr{\"o}bner basis $G$ of the
% ideal $\left< \mbf{f}\right>$ of $\BQ[\mbf{x}, \mbf{y}]$ with respect to the
% elimination order $\texttt{grevlex}(\mbf{x}) \succ \texttt{grevlex}(\mbf{y})$.

\paragraph{Algorithm for computing Hermite matrices}
  \label{sec:parametric_hermite_matrix}
  In~\cite{Le2022RealRootClassification}, an algorithm is described to compute
  the parametric Hermite matrix $\mathcal{H}_1$. Actually this does not only
  compute the matrix $\mathcal{H}_1$ but also the family of matrices $(M_b)_{b
    \in \Basis}$ such that $M_b$ is the matrix of the multiplication map $L_b$
  w.r.t. the basis $\Basis$. We explain now how we can compute $\mathcal{H}_g$
  from $\mathcal{H}_1$. We first compute $\bar{g}$ the normal form of $g$ by the
  Gr{\"o}bner basis $\GB$ (w.r.t. $\texttt{grevlex}(\mbf{x})$ and $\BK$ as a
  base field), namely $\bar{g} = \sum_{b \in \Basis} c_b b$, with $c_b \in
  \BQ(\mbf{y})$. Then we have $M_g = \sum_{b \in \Basis} c_b M_b$, where $M_g$
  denotes the matrix of the multiplication by $g$ in the basis $\Basis$. By
  Proposition~\ref{prop:relation_H1_HQ}, we obtain $\mathcal{H}_g = \mathcal{H}_1 \cdot
  M_g$.

\paragraph*{Specialization properties}
  % \textcolor{blue}{Namely we aim that for a generic choice of some point $\eta
  %   \in \BR^t$, the evaluated matrix $\mathcal{H}_g(\eta)$ coincides with a
  %   Hermite matrix associated to the specialization $(\mbf{f}(\eta,
  %   \cdot),g(\eta, \cdot))$. In particular, by
  %   \cref{thm:signature_tarski_query}, one will be able to compute
  %   Tarski-queries of the form $ \TQ(g(\eta, \cdot), \mbf{f}(\eta, \cdot))$ for
  %   any generic $\eta$ just by computing the signature of an evaluation of the
  %   same Hermite matrix $\mathcal{H}_g$.}
We prove now specialization properties of these parametric Hermite matrices.
The Gr{\"o}bner basis $\GB$ is a subset of $\BQ[\mbf{y},\mbf{x}]$, thus for all $p
\in \GB$, $\lc_{\mbf{x}}(p) \in \BQ[\mbf{y}]$.
We denote by $V(\lc_{\mbf{x}}(p) )$
its vanishing set in $\BC^t$. We consider the following algebraic set
$\mathcal{W}_\infty \subseteq \BC^t$:
\begin{equation}
  \label{eq:W_infinity}
  \mathcal{W}_\infty \coloneqq \bigcup_{p \in \GB} V(\lc_{\mbf{x}}(p)).
\end{equation}

\begin{proposition}
  \label{prop:specialization_hermite_matrix}
  For all $\eta \in \BC^t \setminus \mathcal{W}_\infty$, the specialization
  $\mathcal{H}_g(\eta)$ coincides with the Hermite matrix $\mathcal{H}_g^\eta$
  associated to $(\mbf{f}(\eta,\cdot),g(\eta,\cdot))$ w.r.t. the basis $\Basis$.
\end{proposition}

\begin{proof}
  Let $\eta \in \BC^t \setminus \mathcal{W}_\infty$. By \cite[Lem.
  9]{Le2022RealRootClassification} which is a consequence of
  \cite[Thm.~3.1]{Kalkbrener1997specialization_grobner_basis}, the
  specialization $\GB(\eta,\cdot) \coloneqq \{p(\eta,\cdot) \mid p \in \GB \}$
  is a Gr{\"o}bner basis of the ideal $\left< \mbf{f}(\eta,\cdot)\right>
  \subseteq \BC[\mbf{x}]$ w.r.t. the ordering $\texttt{grevlex}(\mbf{x})$. Since
  $\eta \in \BC^t \setminus \mathcal{W}_\infty$, the leading coefficient
  $\lc_{\mbf{x}}(p)$ does not vanish at $\eta$ for all $p \in \GB$. Thus, the
  set of leading monomials of $\GB$ in the variables $\mbf{x}$ w.r.t.
  $\texttt{grevlex}(\mbf{x})$ is exactly the set of leading monomials of
  $\GB(\eta,\cdot)$ w.r.t. $\texttt{grevlex}(\mbf{x})$. Therefore, the finite
  set $\Basis$ is also the set of monomials in $\mbf{x}$ that are not reducible
  by $\GB(\eta, \cdot)$. Hence $\Basis$ is a basis of the quotient ring
  $\BC[\mbf{x}]/\left< \mbf{f}(\eta, \cdot) \right>$. So, $\left< \mbf{f}(\eta,
    \cdot) \right>$ is zero-dimensional and one can define $\mathcal{H}_g^\eta$
  as the Hermite matrix associated to $(\mbf{f}(\eta,\cdot),g(\eta,\cdot))$
  w.r.t. the basis $\Basis$.

  Moreover, the specialization property for the Gr{\"o}bner basis $\GB$ implies
  that when dividing some polynomial w.r.t. $\GB$, none of the denominators
  which appear vanish at $\eta \in \BC^{t}$. Hence, given $h\in
  \BQ[\mbf{y},\mbf{x}]$ and its normal form $\bar{h}$ w.r.t. $\GB$ (computed
  with $\BK$ as a base field and w.r.t. $\texttt{grevlex}(\mbf{x})$), for any
  $\eta \in \BC^{t}\setminus \mathcal{W}_{\infty}$, $\bar{h}(\eta, \cdot)$
  coincides with the normal form of $h(\eta, \cdot)$ w.r.t. $\GB(\eta, \cdot)$. This
  implies that $\mathcal{H}_g(\eta)= \mathcal{H}_g^{\eta}$.
\end{proof}

\begin{corollary}
  \label{cor:specialized_hermite_matrix_tarski_queries}
  For all $\eta \in \BC^t \setminus \mathcal{W}_\infty$, the signature of
  $\mathcal{H}_g(\eta)$ is the Tarski-query of $g(\eta,\cdot)$ for the
  zero-dimensional system $\mbf{f}(\eta,\cdot)$.
\end{corollary}

\begin{proof}
  By Proposition~\ref{prop:specialization_hermite_matrix}, $\mathcal{H}_g(\eta)$ is a
  Hermite matrix associated to $(\mbf{f}(\eta,\cdot),g(\eta,\cdot))$. The result
  follows from Theorem~\ref{thm:signature_tarski_query} .
\end{proof}

Note that $\mathcal{W}_\infty$ does not depend on $g$. One can compute the
number of real roots of the specialized system $\mbf{f}(\eta,\cdot)$ satisfying
some sign condition for $g$ by computing the signatures of three parametric
Hermite matrices evaluated in $\eta$ and inverting the
system~(\ref{eq:sign_determination_system_1_poly}).

\subsection{Degree bounds}

We bound the degrees of the entries of the parametric Hermite matrix
$\mathcal{H}_g$ under some assumptions that we make explicit below. We start by
recalling the definition of a regular sequence.

\begin{definition}[Regular sequence]
  \label{def:affine_regular_sequence}
  Let $(f_1, \dots, f_e) \subset \BK[\mbf{x}]$ with $e \le n$ be a homogeneous
  polynomial sequence. We say that $(f_1, \dots, f_e)$ is a \emph{homogeneous
    regular sequence} if for all $1 \le i \le e $, $f_i$ is not a zero-divisor
  in $\BK[\mbf{x}]/\left< f_1, \dots, f_{i-1}\right>$.

  A polynomial sequence $(f_1, \dots, f_e) \subset \BK[\mbf{x}]$ is called an \emph{affine regular sequence} if
  $(f_1^H, \dots, f_e ^H)$ is a homogeneous regular sequence, where for a polynomial $q \in \BK[\mbf{x}]$,
  $q^H$ denotes the homogeneous component of largest degree of $q$.
\end{definition}

First we bound the degrees of the entries of the matrix
$\mathcal{H}_g$. For $p \in \BQ[\mbf{y}][\mbf{x}]$ we denote by $\deg(p)$ the
total degree of $p$ and $\deg_{\mbf{x}}(p)$ (resp. $\deg_{\mbf{y}}(p)$) the
degree of $p$ w.r.t. $\mbf{x}$ (resp. $\mbf{y}$). Let $d \coloneqq \max_{1 \le i
  \le p} \deg(f_i)$. We consider the reduced Gr{\"o}bner basis $\GB$ as above
and the associated monomial basis $\Basis$ of monomials in
$\mbf{x}$ of the finite dimensional vector space $\AK$. The quotient ring $\AK$
has dimension $\delta$. Note that if $\mbf f$ is an afiine regular sequence, the codimension
of the ideal generated by $f_1, \ldots, f_e$ is $e$ if this ideal is not
$\langle 1 \rangle$. Hence, combined with
Assumption~\ref{assump:finite_nb_sol_generically}, this forces $e=n$ and by B{\'e}zout's
inequality, we have $\delta \le d^n$. We recall below
Assumption~\ref{assump:generic_degree_grobner_basis}.

%\begin{assumption}
%  \label{assump:generic_degree_grobner_basis}
 \emph{ For any $p \in \GB$, we have $\deg p = \deg_{\mbf{x}} p $.}
%\end{assumption}

 \smallskip
\noindent By \cite[Prop.~20]{Le2022RealRootClassification},
Assumption~\ref{assump:generic_degree_grobner_basis} holds for \emph{generic} sequences
$\mbf{f}$. 
%\smallskip

% \begin{proposition} [{\cite[Prop.~20]{Le2022RealRootClassification}}]
%   \label{prop:genericity_assump_degree_grobner_basis}
%   Let $\BC[\mbf{x},\mbf{y}]_d$ be the set of polynomials in $\BC[\mbf{x},\mbf{y}]$
%   having total degree less than $d$.
%   There exists a nonempty Zariski open subset $\mathcal{F}$ of $\BC[\mbf{x}, \mbf{y}]_{d}^n$
%   such that \cref{assump:generic_degree_grobner_basis} holds
%   for any $\mbf{f} \in \mathcal{F} \cap \BQ[\mbf{x}, \mbf{y}]^n$. 
% \end{proposition}

%\noindent The following are consequences of~\cref{assump:generic_degree_grobner_basis}.

\begin{lemma}
  \label{lem:leading_coefficient_grobner_basis}
  Under Assumption~\ref{assump:generic_degree_grobner_basis},
  $\deg_{\mbf{y}}(\lc_{\mbf{x}}(p)) = 0$ for all $p \in \GB$.
  % If \cref{assump:generic_degree_grobner_basis} holds, then for all $p \in G$,
  % $\lc_{\mbf{x}}(p)$ has degree $0$ in $\mbf{y}$.
\end{lemma}

\begin{proof}
  Let $p \in \GB$, by definition of the ordering $\texttt{grevlex}(\mbf{x})
  \succ \texttt{grevlex}(\mbf{y})$, $\lc_{\mbf{x}}(p)$ is obtained from a term  $\tau$ in $p$
  s.t. $\deg_{\mbf{x}}(\tau) = \deg_{\mbf{x}}(p)$. By
  Assumption~\ref{assump:generic_degree_grobner_basis}, $\deg_{\mbf{x}}(p) = \deg(p)$, so
  $\deg_{\mbf{y}}(\tau) = 0$.
\end{proof}

\begin{lemma}
  \label{lem:degree_bound_normal_form}
   If Assumption~\ref{assump:generic_degree_grobner_basis} holds, then for any $q \in
  \BQ[\mbf{y}][\mbf{x}]$, the normal form $\bar{q}$ of $q$ w.r.t. $\GB$ lies in
  $\BQ[\mbf{y}][\mbf{x}]$ and $\deg(\bar{q}) \le \deg(q)$.
\end{lemma}

\begin{proof}
  Let $q \in \BQ[\mbf{y}][\mbf{x}]$, $\bar{q}$ is the remainder of successive
  divisions of $q$ by polynomials in $\GB$. As
  Assumption~\ref{assump:generic_degree_grobner_basis} holds, by
  Lemma~\ref{lem:leading_coefficient_grobner_basis}, those divisions do not introduce
  any denominator. So, every term appearing during these reductions are
  polynomials in $\BQ[\mbf{y}][\mbf{x}]$. By
  Assumption~\ref{assump:generic_degree_grobner_basis}, for any $p \in \GB$, the total
  degree of every term of $p$ is bounded by $\deg_{\mbf{x}}(p) =
  \deg(\lm_{\mbf{x}}(p))$ by Lemma~\ref{lem:leading_coefficient_grobner_basis}. Thus,
  a division of $q$ by $p$ involves only terms of total degree $\deg(q)$.
  Therefore, during the normal form reduction of $q$ by $\GB$, only terms of
  degree at most $\deg(q)$ will appear. Hence $\deg(\bar{q}) \le \deg(q)$.
\end{proof}

% \begin{lemma}
  % \label{lem:hermite_matrix_polynomial_entries}
  % If \cref{assump:generic_degree_grobner_basis} holds, then for all $g \in \BQ[\mbf{y}][\mbf{x}]$
  % the entries of the parametric Hermite matrix $\mathcal{H}_g$ associated to $(\mbf{f},g)$
  % in the basis $B$ are polynomials in $\mbf{y}$,
  % \ie $\mathcal{H}_g  \in \BQ[\mbf{y}]^{\delta \times  \delta}$.
% \end{lemma}

% \begin{proof}
%   The entry $(i,j)$ of $\mathcal{H}_g$ is defined as $\tr(L_{g b_ib_j})$ and
%   by \cref{lem:degree_bound_normal_form}
%   the reduction in $A_\BK$ of any polynomial in $\BQ[\mbf{y}][\mbf{x}]$ is in $\BQ[\mbf{y}][\mbf{x}]$.
%   So for all $1 \le k \le \delta$, the normal form of $g b_ib_jb_k$ is a linear combination of elements of $B$
%   with coefficients in $\BQ[\mbf{y}]$ and
%   the matrix of $L_{g b_ib_j}$ in the basis $B$ has its entries in $\BQ{[\mbf{y}]}$.
%   Hence $\tr(L_{g b_ib_j}) \in \BQ[\mbf{y}]$. 
% \end{proof}

\noindent We prove now degree bounds on the entries of $L_g$ and
$\mathcal{H}_g$.
% Now we can show that the entries of the matrix of $L_g$ and of
% $\mathcal{H}_g$ are polynomials in $\mbf{y}$ and we can bound their degrees.
\begin{lemma}
  \label{lem:degree_bound_multiplication_matrix}
  Under Assumption~\ref{assump:generic_degree_grobner_basis}, let $g \in \BQ[\mbf{y}][\mbf{x}]$ and
  let us denote by $(g_{i,j})_{1 \le i,j \le \delta}$ the matrix of $L_g$ in the basis $\Basis$.
  Then, for all $1 \le i,j \le \delta$, $g_{i,j} \in \BQ[\mbf{y}]$
  and $\deg(g_{i,j}) \le \deg(g) + \deg(b_j) - \deg(b_i)$.
\end{lemma}

\begin{proof}
  We have for all $1 \le j \le \delta $, $\bar{g b_j} = \sum_{i=1}^\delta
  g_{i,j} b_i$ and $g_{i,j} \in \BQ[\mbf{y}]$
  by Lemma~\ref{lem:degree_bound_normal_form}. Also, $\deg(\bar{gb_j}) = \max
  \deg(g_{i,j}b_i)$ because $g_{i,j} \in \BQ[\mbf{y}]$ and the $b_i$'s are
  distinct monomials in $\mbf{x}$. In particular, for all $1 \le i \le \delta$,
  $\deg(g_{i,j} b_i) = \deg(g_{i,j}) + \deg(b_i) \le \deg(\bar{gb_j}) \le
  \deg(gb_j)\leq \deg(g) + \deg(b_j) $ by Lemma~\ref{lem:degree_bound_normal_form}.
  The result follows.
\end{proof}

\begin{proposition}
  \label{prop:degree_bound_hermite_matrix}
  Under Assumption~\ref{assump:generic_degree_grobner_basis}, let $g \in \BQ[\mbf{y}][\mbf{x}]$ and
  let us denote by $(h_{i,j})_{1 \le i,j \le \delta}$
  the Hermite matrix $\mathcal{H}_g$ associated to $(\mbf{f},g)$ in the basis $\Basis$.
  Then,  for all $1 \le i,j \le \delta$, $h_{i,j} \in \BQ[\mbf{y}]$
  and \begin{align*} \deg(h_{i,j}) \le \deg(g) + \deg(b_i) + \deg(b_j). \end{align*}
  As a direct consequence, the degree of a minor of $\mathcal{H}_g$ defined by the rows
  $(r_1, \dots, r_k)$ and the columns $(c_1, \dots, c_k)$ is bounded by
  \begin{align*} k\deg(g) + \sum_{i=1}^k (\deg(b_{r_i}) + \deg(b_{c_i})). \end{align*}
\end{proposition}

\noindent Hence, the determinant of $\mathcal{H}_g$ has degree bounded by
$\delta \deg(g) + 2\sum_{i=1}^{\delta} \deg(b_i)$ and $\delta \le d^n$.
\begin{proof}
  We have $h_{i,j} = \tr(L_{gb_i b_j})$. Let $C = (c_{k,\ell})_{1 \le k,\ell \le
  \delta }$ denote the entries of the matrix of $L_{gb_i b_j}$ w.r.t. $\Basis$. By
  Lemma~\ref{lem:degree_bound_multiplication_matrix}, as $gb_i b_j \in
  \BQ[\mbf{y}][\mbf{x}]$, $C \in \BQ[\mbf{y}]^{\delta \times \delta}$. So,
  $h_{i,j} = \tr(C) = \sum_{k=1}^\delta c_{k,k}\in \BQ[\mbf{y}]$. By
  Lemma~\ref{lem:degree_bound_multiplication_matrix}, $\deg(h_{i,j}) \le \max_k
  \deg(c_{k,k}) \le \max_k \deg(gb_ib_j)$. Therefore, $\deg(h_{i,j}) \le \deg(g)
  + \deg(b_i) + \deg(b_j)$.

  The degree bound for the minors of $\mathcal{H}_g$ is clear by expanding the
  expression for determinants.
\end{proof}

% Next if we also suppose that $\mbf{f}$ is an affine regular sequence, we know explicit degree bounds
% for the elements of the basis $b_i$ that we can apply to bound the degrees of any minor of $\mathcal{H}_g$.

% \begin{lemma}[{\cite[Lem.~23]{Le2022RealRootClassification}}]
%   \label{lem:degree_bound_basis_regular_sequence}
%   If $\mbf{f}$ is an affine regular sequence, then the highest degree among the elements of $B$ is bounded by
%   $n(d-1)$ and
%   \begin{align*}
%     2 \sum_{i=1}^\delta \deg(b_i) \le n(d-1)d^n.
%   \end{align*}
% \end{lemma}

% \noindent Combining the bounds of \cref{lem:degree_bound_basis_regular_sequence} and
% \cref{prop:degree_bound_hermite_matrix}, we obtain the following result.

\begin{corollary}
  \label{cor:degree_minor_hermit_matrix_regular}
  Assume that $\mbf{f}$ is an affine regular sequence satisfying Assumption~\ref{assump:generic_degree_grobner_basis}.
  For $g \in \BQ[\mbf{y}][\mbf{x}]$,
  the degree of any minor of $\mathcal{H}_g$ is bounded by $(\deg(g) + n(d-1))d^n$.
\end{corollary}

\begin{proof}
  Since $\mbf{f}$ is an affine regular sequence, one can apply
  \cite[Lem.~23]{Le2022RealRootClassification}. Hence, the highest degree among
  the elements of $\Basis$ is bounded by $n(d-1)$ and it holds that $2
  \sum_{i=1}^\delta \deg(b_i) \le n(d-1)d^n$. Substituting these degree bounds
  on the $b_i$'s in the ones of Proposition~\ref{prop:degree_bound_hermite_matrix} ends the
  proof.
\end{proof}

\section{Algorithm}\label{sec:algo}

%\subsection{Realizable sign conditions on a real algebraic set}
%\label{sec:real_sign_cond_alg_set}

Let $\mathcal{Q} = (Q_1, \dots, Q_s)\subset \BQ[\mbf{X}]$ for $\mbf{X} =
(X_1,\dots, X_k)$. Let $\mathcal{Z}$ be a subset of $\BR^k$. We say that an
element $\sigma \in \{0,1,-1\}^s$ is a sign condition for $\mathcal{Q}$. A sign
condition $\sigma$ for $\mathcal{Q}$ is said to be realizable over $\mathcal{Z}$
if the following set is nonempty
\begin{align*}
  \Reali(\sigma, \mathcal{Z}) \coloneqq \{ x \in  \mathcal{Z} \mid
  \bigwedge_{i =1}^s \sign(Q_i(x)) = \sigma(i) \}.
\end{align*}
We consider the set $\SIGN(\mathcal{Q},\mathcal{Z}) \subseteq \{0,1,-1\}^s$ of
realizable sign conditions for $\mathcal{Q}$ over $\mathcal{Z}$.

%Let $\mathcal{P}= (P_1, \dots, P_\ell)$ and 
%be two finite families of 
%Typically, one can think to $\mbf{f}$ and $\mbf{g}$ as an instance of \cref{pb:solution_classification}.
%Let $\mathcal{Z} \subseteq \BR^{t+n}$ be the real algebraic set defined by $\mathcal{P} = 0$.
%

Let $\mbf{f}$ and $\mbf{g}$ be an instance of Problem~\ref{pb:solution_classification}.
We are interested in the set of sign conditions $\SIGN(\mbf{g},
\mathcal{V}_\BR)$. We describe an algorithm for the determination of the
realizable sign conditions for $\mbf{g}$ over the real solutions of $\mbf{f} = 0$
when $\mbf{f}$ satisfies Assumption~\ref{assump:finite_nb_sol_generically}. This algorithm
does not recover the whole set $\SIGN(\mbf{g}, \mathcal{V}_\BR)$, but only the
set of realizable conditions for $\mbf{g}$ over $\mathcal{V}_\BR \cap
\mathcal{W}$ where $ \mathcal{W}$ is a nonempty Zariski open subset of
$\BC^{t+n}$. It means that we potentially miss some elements of $\SIGN(\mbf{g},
\mathcal{V}_\BR)$ but they can only occur for $(\mbf{y}, \mbf{x})$ lying in a
Zariski closed set. Also, as a by-product our algorithm enables us to compute a
valid solution to Problem~\ref{pb:solution_classification} for $(\mbf{f}, \mbf{g})$.

This algorithm is a variant of \cite[Chap.~10]{roy2006algo_real_alg_geom} for
determining the sign conditions realized by a family of polynomials on a finite
set of points in $\BR^k$ using Tarski-queries. Tarski-queries are expressed as
signatures of parametric Hermite matrices as
in Corrolary~\ref{cor:specialized_hermite_matrix_tarski_queries}. We also use sample
points algorithms as
in~\cite{Le2022RealRootClassification,SafeySchost2003OnePointPerConnectedComponent}.
In the case where $\mbf{g}$ is empty, this algorithm coincides with the one
in~\cite{Le2022RealRootClassification}.

\subsection{Sign determination on a finite set of points}

We are given a family of polynomials $\mathcal{Q} = (Q_1(\mbf{X}), \dots,
Q_s(\mbf{X})) \subset \BQ[\mbf{X}]$ with $\mbf{X} = (X_1, \dots, X_k)$ and an
\emph{implicit} finite set of points $Z$ in $\BR^k$ of size $\rho$. By implicit we
mean that we have no explicit description for the points in $Z$. Typically, the
set $Z$ designates the roots of the system $\mbf{f}(\eta,\cdot) = 0$ with $\eta$
in the space of parameters and $\mathcal{Q}$ is the family
$\mbf{g}(\eta,\cdot)$. For now we assume that we have access to a black-box for
computing the Tarski-queries $\TQ(Q, Z)$ for any $Q \in \BQ[\mbf{X}]$.
%For this black-box, one can think to an algorithm that computes the signature of an evaluation of a parametric
%Hermite matrix as in~\cref{cor:specialized_hermite_matrix_tarski_queries}.
We aim at computing $ \SIGN(\mathcal{Q}, Z)$.

For a sign condition $\sigma \in \{0,1,-1\}^s$ for $\mathcal{Q}$, define
$c(\sigma, Z) \coloneqq \sharp \Reali(\sigma, Z)$ and, for $\alpha \in \{0,1,2\}^s$,
denote $ \mathcal{Q}^\alpha \coloneqq \prod_{i=1}^s Q_i^{\alpha(i)}$,
and % \quad
$ \sigma^\alpha \coloneqq \prod_{i=1}^s \sigma(i)^{\alpha(i)}$. One can notice
that on $\Reali(\sigma,Z)$, the sign of $\mathcal{Q}^{\alpha}$ is fixed and
equal to $ \sigma^\alpha$.

We also order $\{0,1,-1\}^s$ with the lexicographic order induced by $0 < 1 <
-1$ and $\{0,1,2\}^s$ with the lexicographic order induced by $0 < 1 < 2$. Let
$\Sigma = \{ \sigma_1, \dots , \sigma_{\nu} \} \subset \{0,1,-1\}^s$ with $\sigma_1
<_\lex \dots <_\lex \sigma_\nu$, we denote by $c(\Sigma,Z)$ the column vector
$(c(\sigma_1,Z), \dots,c(\sigma_\nu,Z))^t$. Similarly, let $A = \{\alpha_1, \dots,
\alpha_\gamma\} \subset \{0,1,2\}^s$ with $\alpha_1 <_\lex \dots <_\lex \alpha_\gamma$, we
denote by $\TQ(\mathcal{Q}^A, Z)$ the column vector $
\left(\TQ(\mathcal{Q}^{\alpha_1},Z), \dots,
  \TQ(\mathcal{Q}^{\alpha_\gamma},Z)\right)^t$.

% \begin{definition} [Matrix of signs]
%   \label{def:matrix_of_signs}
  We define the matrix of signs of $A$ on $\Sigma$ as the $\gamma \times \nu$ matrix
  $\Mat(A,\Sigma)$ whose entry $(i,j)$ is $\sigma_j^{\alpha_i}$.
%\end{definition}

\begin{proposition}[{\cite[Prop. 10.59]{roy2006algo_real_alg_geom}}]
  \label{prop:sign_determination_general_system}
  If $\SIGN(\mathcal{Q},Z) \subseteq \Sigma$, then it holds that
  %\begin{equation*}
    $\Mat(A, \Sigma) \cdot c(\Sigma, Z)= \TQ(\mathcal{Q}^A, Z)$. 
  %\end{equation*}
\end{proposition}

% \textcolor{magenta}{Proof commented out.}
% \begin{proof}
%   Since $\SIGN(\mathcal{Q},Z) \subseteq \Sigma$,
%   we can write $Z$ as the disjoint union $\bigcup_{k=1}^p \Reali(\sigma_k,Z) = Z$.
%   Thus, for $1  \le i \le m$, we have
%   \begin{align*}
%     \TQ(\mathcal{Q}^{\alpha_i}, Z)
%     &= \sum_{x \in Z} \sign(\mathcal{Q}^{\alpha_i}(x))
%       = \sum_{k=1}^p \sum_{x \in \Reali(\sigma_k, Z)}
%       \underbrace{\sign(\mathcal{Q}^{\alpha_i}(x))}_{= \sigma^{\alpha_i}_k}
%       = \sum_{k=1}^p \sigma_k^{\alpha_i}c(\sigma_k,Z),
%   \end{align*}
%   whence the result.
% \end{proof}

Hence, when $\Mat(A,\Sigma)$ is invertible, one can determine $c(\Sigma,Z)$
and thus $\SIGN(\mathcal{Q},Z)=\{\sigma \in \Sigma \mid c(\sigma,Z) > 0\}$ from
$\TQ(\mathcal{Q}^A,Z)$ by linear system solving. In this case, we say that the
set $A$ is adapted to $\Sigma$ for sign determination. If one chooses $\Sigma =
\{0,1,-1\}^s$ to be the whole set of possible sign conditions for $\mathcal{Q}$,
the set $A = \{0,1,2\}^s$
is adapted to $\Sigma$~\cite[Prop. 10.60]{roy2006algo_real_alg_geom}.
However, we need to compute $3^s$ Tarski-queries to perform sign determination.
Yet the number of realizable
sign conditions is bounded by the number $\rho$ of elements in $Z$ and often $\rho \ll
3^s$. So when $\Sigma = \{0,1, -1\}^s$, many entries in $c(\Sigma,Z)$ are equal
to 0. To avoid to compute an exponential number of Tarski-queries, we want to
avoid unrealizable sign conditions. To do so, we use the incremental approach of
\cite[Sec.~10.3]{roy2006algo_real_alg_geom}. Let $\mathcal{Q}_i \coloneqq
(Q_{s-i + 1}, \dots, Q_s )$ be the last $i$ polynomials in $\mathcal{Q}$. At
step $i$, we compute $\SIGN(\mathcal{Q}_i,Z)$ the realizable sign conditions for
$\mathcal{Q}_i$, for $i$ from 1 to $s$, so that we get rid of the empty sign
conditions at each step of the computation.

First for any $\Sigma \subseteq \{0,1,-1\}^s$, we explain how to construct a set
$A \subseteq \{0,1,2\}^s$ that is adapted to $\Sigma$. For $\sigma =(\sigma(1),
\dots, \sigma(s))\in \{0,1,-1\}^s$, we denote by $\sigma'$ the vector obtained
by removing the first coordinate of $\sigma$, \ie $\sigma' \coloneqq (\sigma(2),
\dots, \sigma(s))\in \{0,1,-1\}^{s-1}$.
%For $\tau \in \{0,1,-1\}$, we also
%denote $\tau \land \sigma \coloneqq (\tau, \sigma(1), \dots, \sigma(s))$.
\begin{definition}
  For $\Sigma \subseteq \{0,1,-1\}^s$, we define
   $ \Sigma_1' \coloneqq \{\sigma' \mid \sigma \in \Sigma \}$,
   and the subsets $\Sigma'_2, \Sigma'_3 \subseteq \Sigma'_1$ such that
   $\Sigma'_2$ (resp. $\Sigma'_3$) contains the elements of $\Sigma'_1$
   that can be extended to an element of
  $\Sigma$ in at least two different ways (resp. exactly three different ways).
\end{definition}

\begin{definition}
  \label{def:adapted_family}
  Let $\Sigma \subseteq \{0,1,-1\}^s$,
  we define $\Ada(\Sigma) \subseteq \{0,1,2\}^s$ by induction on $s \ge 1$ as follows:
  \begin{itemize}[noitemsep,label=$-$]
  \item if $s= 1$, let $h \in \{1,2,3\}$ be the size of $\Sigma$, and set $\Ada(\Sigma) = \{0, \dots, h-1\}$;
  \item if $s > 1$, $\Ada(\Sigma) = 0 \times \Ada(\Sigma_1') \cup  1 \times \Ada(\Sigma_2') \cup
    2 \times \Ada(\Sigma_3')  $.
  \end{itemize}
\end{definition}

\begin{proposition}[{\cite[Prop. 10.65]{roy2006algo_real_alg_geom}}]
  \label{prop:adapted_family}
  The set $\Ada(\Sigma)$ is adapted to $\Sigma$ for sign determination.
\end{proposition}

Now suppose that for $1 \le i < s$, we have built $ \SIGN(\mathcal{Q}_i,Z)$ and $\Ada( \SIGN(\mathcal{Q}_i,Z))$,
we explain how we can compute $\SIGN(\mathcal{Q}_{i+1},Z)$ and $\Ada(\SIGN(\mathcal{Q}_{i+1},Z))$.
It is based on the two following lemmas.

\begin{lemma}
  \label{lem:matrix_of_signs_cartesian_product}
  Let $s_1, s_2 \ge 0$, and $A_1 \subseteq \{0,1,2\}^{s_1}$, $A_2 \subseteq \{0,1,2\}^{s_2}$,
  $\Sigma_1 \subseteq \{0,1,-1\}^{s_1}$, $\Sigma_2 \subseteq \{0,1,-1\}^{s_2}$.
  The matrix of signs of $A_1 \times A_2$ on $\Sigma_1 \times \Sigma_2$ is
  %\begin{equation*}
  $  \Mat(A_1 \times A_2, \Sigma_1 \times \Sigma_2) =
    \Mat(A_1, \Sigma_1 ) \otimes \Mat( A_2,\Sigma_2)$.
  %\end{equation*}
  As a consequence, if $A_1$ is adapted to $\Sigma_1$ and $A_2$ is adapted to $\Sigma_2$ then $A_1 \times A_2$
  is adapted to $\Sigma_1 \times \Sigma_2$.
\end{lemma}

\begin{proof}
  Let $\alpha =(\alpha_1, \alpha_2) \in A_1 \times A_2$ and $\sigma =(\sigma_1,
  \sigma_2) \in \Sigma_1 \times \Sigma_2$. By definition, we have $\sigma^\alpha
  = \sigma_1^{\alpha_1} \cdot \sigma_2^{\alpha_2}$. Since rows and columns of
  matrices of signs are ordered with the lexicographic orderings induced by $0 <
  1 < 2$ for the rows and $0 < 1 < -1$ for the columns, the result holds.
\end{proof}

\begin{lemma}[{\cite[Lem. 10.66]{roy2006algo_real_alg_geom}}]
  \label{lem:adapted_family_subset}
  Let $ \Sigma \subseteq \Gamma \subseteq \{0,1,-1\}^s$ and $\nu = \sharp \Sigma$.
  The matrix $\Mat(\Ada(\Sigma),\Sigma)$ is the matrix obtained by extracting the
  first $\nu$ linearly independent rows of $\Mat(\Ada(\Gamma),\Sigma)$, \ie the rows
  corresponding to the row rank profile of $\Mat(\Ada(\Gamma),\Sigma)$.
\end{lemma}

The computation of $\SIGN(\mathcal{Q}_{i+1},Z)$ and
$\Ada(\SIGN(\mathcal{Q}_{i+1},Z))$ is described
in Algorithm~\ref{alg:sign_determination_one_step}.
% We start by computing $S = \Sign(Q_{s-i+1},Z)$ and $A = Ada(S)$.
% Since $A \times \Ada(\SIGN(\mathcal{Q}_{i+1},Z))$ is adapted to $S \times \SIGN(\mathcal{Q}_{i+1},Z)$
% for sign determination, and
% $ \Sign(\mathcal{Q_i},Z) \subseteq S \times \SIGN(\mathcal{Q}_{i+1},Z)$
% one can deduce $\Sign(\mathcal{Q_i},Z)$ by inverting a linear system.
% Finally one can compute $\Ada(\Sign(\mathcal{Q_i},Z))$ from the row rank profile
After $s$ iterations of this algorithm we get $\SIGN(\mathcal{Q},Z)$.
This is \cite[Alg.~10.11]{roy2006algo_real_alg_geom}.

\begin{algorithm}[t]
  \caption{One step of sign determination}
  \label{alg:sign_determination_one_step} 
  \Input{ $\mathcal{Q} = \{Q\} \cup \mathcal{Q}'$,
    the sets $\Sigma \coloneqq \SIGN( \mathcal{Q}',Z)$ and $\Ada(\Sigma)$,
    the associated matrix of signs $\Mat(\Ada(\Sigma),\Sigma)$}
  \Output{The sets $\SIGN(\mathcal{Q}, Z), \Ada(\SIGN(\mathcal{Q}, Z))$ and the associated matrix of signs}
  %\Oracle{For a polynomial $Q$, the Tarski-query $\TQ(Q,Z)$}
  Compute $S \coloneqq \SIGN({Q},Z)$ from the Tarski-queries $\TQ(1,Z), \TQ(Q,Z), \TQ(Q^2,Z)$
  by solving~(\ref{eq:sign_determination_system_1_poly}).
   $S$ corresponds to the nonzero entries of the solution\\
Deduce $A \coloneqq \Ada(S)$ from Definition~\ref{def:adapted_family} \\
Compute the vector % of Tarski-queries
$T \coloneqq \TQ(\mathcal{Q}^{A\times \Ada(\Sigma)},Z)$ \\
$M \leftarrow \Mat(A \times \Ada(\Sigma), S \times \Sigma)
=    \Mat(A,S) \otimes \Mat(\Ada(\Sigma), \Sigma)$ \\
Compute $c \coloneqq c(S \times \Sigma, Z)$ by solving $   M \cdot c = T$\\
Deduce  $\SIGN(\mathcal{Q},Z)$ \# \emph{given by the nonzero entries in $c$}\\
Delete in $M$ the columns whose index is not in $\SIGN(\mathcal{Q},Z)$ \\
Deduce $\Ada(\SIGN(\mathcal{Q},Z))$ from the row rank profile of $M$ and delete the other rows.
\end{algorithm}

\subsection{General sign determination}

We design an algorithm based on Algorithm~\ref{alg:sign_determination_one_step} to solve
Problem~\ref{pb:solution_classification}.

Let $\mbf{f} = (f_1, \dots , f_e ) \subset \BQ[\mbf{y}][\mbf{x}]$ such that
$\mbf{f}$ satisfies Assumption~\ref{assump:finite_nb_sol_generically} and let $\mbf{g} =
(g_1, \dots, g_s) \subset \BQ[\mbf{y}][\mbf{x}]$ define the inequalities of our
input system.
As before, let $\GB$ be the reduced Gr{\"o}bner basis of
$\left< \mbf{f} \right> $ w.r.t. the ordering
$\texttt{grevlex}(\mbf{x}) \succ \texttt{grevlex}(\mbf{y})$.  We also
denote by $\BK$ the field $\BQ(\mbf{y})$ and $\Basis \subset \BQ[\mbf{x}]$
is the basis of
$\AK \coloneqq \BK[\mbf{x}] / \left< \mbf{f}\right>_\BK$ derived from
$\GB$ of dimension $\delta$.

Let $g \in \BQ[\mbf{y}][\mbf{x}]$,
$\mathcal{H}_g \in \BK^{\delta \times \delta}$ denotes the Hermite
matrix associated to $(\mbf{f},g)$ in $\Basis$.
We consider as in~(\ref{eq:W_infinity}), the algebraic set
$\mathcal{W}_\infty = \cup_{p \in \GB} V(\lc_{\mbf{x}}(p)) \subset
\BC^t$. %  for which the parametric matrices $\mathcal{H}_g$ have good
% specialization property
% (\cref{cor:specialized_hermite_matrix_tarski_queries}) when $\eta$
% ranges outside $\mathcal{W}_\infty$.
%
\begin{lemma}
  \label{lem:generically_nonzero_leading_principal_minors}
  Let $r$ denotes the rank of $\mathcal{H}_g$. There exists a Zariski dense
  subset $\mathcal{U}_g$ of $ \GL_\delta(\BC)$ such that for $U \in
  \mathcal{U}_g$, the first $r$ leading principal minors of $\mathcal{H}_g^U
  \coloneqq U^t \mathcal{H}_g U$ are not identically zero.
\end{lemma}
\begin{proof}
  The matrix $\mathcal{H}_g$ has rank $r$ so there exists
  $\eta \in  \BR^t \setminus \mathcal{W}_\infty$ 
  such that the evaluation $\mathcal{H}_g(\eta)$
  is a matrix of rank $r$.
  % of $\mathcal{H}_g$ of size $r$ that is invertible. Then $\det(H)$ is not
  % identically 0. Let $\mathcal{W}_{g} \subset \BC^t$ be the vanishing set of the
  % numerator of $\det(H)$. 
  %The algebraic set $\mathcal{W}_{g}$ \todo{What is $\mathcal{W}_g$?} is a proper Zariski closed subset of
  %$\BC^t$. Let $\eta \in \BR^t \setminus (\mathcal{W}_\infty \cup
  %\mathcal{W}_g)$. The evaluation $\mathcal{H}_g(\eta)$ is
  %\textcolor{magenta}{allowed} and is a real symmetric matrix of rank $r$.
  Moreover, for all $U \in \GL_\delta(\BC)$, $\mathcal{H}_g^U(\eta) = U^t
  \mathcal{H}_g(\eta) U$. We show that there exists a Zariski dense subset
  $\mathcal{U}_g$ such that for all $U \in \mathcal{U}_g$, the first $r$ leading
  principal minors of $\mathcal{H}_g^U(\eta)$ are nonzero. This would imply that
  the first $r$ leading principal minors of $\mathcal{H}_g^U$ are not
  identically zero.

  For $1 \le j \le r$ and $R,C \subset \{1,\dots,r\}$ of size $j$,
  we denote by $\mathcal M_{R,C}$ the $j \times j$ minor of $\mathcal{H}_g(\eta)$
  corresponding to the sets of rows $R$ and of columns $C$.
  %, let us denote by $\mathcal{M}_j$ the set of all
  %$j \times j$ minors of $\mathcal{H}_g(\eta)$.
  We consider the
  matrix $U \coloneqq (\mathfrak{u}_{i,j})_{1 \le i,j \le \delta}$
  where $\mathfrak{u} = (\mathfrak{u}_{i,j})$ are new indeterminates.
  Then, the $j$-th leading principal minor $M_j(\mathfrak{u})$ of the
  matrix $\mathcal{H}_g^U(\eta)$ can be written as
  \begin{equation*}
    M_j(\mathfrak{u}) = \sum_{R,C \subset \{1,\dots,r\}:~ \left|R \right| = \left|C \right| = j} \mathcal P_{R,C} \cdot \mathcal M_{R,C},
  \end{equation*}
  where the $\mathcal P_{R,C}$'s are elements of $\BQ[\mathfrak{u}]$.  As
  $\mathcal{H}_g(\eta)$ is a real symmetric matrix of rank $r$ there
  exists a matrix $Q \in \GL_\delta(\BR)$ such that
  \begin{equation*}
    \mathcal{H}_g^Q(\eta) = Q^t \mathcal{H}_g(\eta) Q = \begin{bmatrix} \Delta & 0 \\ 0 & 0 \end{bmatrix},
  \end{equation*}
  where $\Delta$ is a diagonal matrix of size $r$ with nonzero real
  entries on its diagonal.  Hence the evaluation of $\mathfrak{u}$ at
  the entries of $Q$ gives $M_j(\mathfrak{u})$ a nonzero value.  So
  we conclude that $M_j(\mathfrak{u})$ is not identically zero.

  Finally, let $\mathcal{U}_j$ be the nonempty Zariski open subset of
  $\GL_\delta(\BC)$ defined as the non-vanishing set of $M_j(\mathfrak{u})$. We
  define $\mathcal{U}_g$ as the intersection of $\mathcal{U}_j$ for $1 \le j \le
  r$, and for $U \in \mathcal{U}_g$, none of the first $r$ leading principal
  minors of $ \mathcal{H}_g^U(\eta)$ is zero. Thus, none of the first $r$
  leading principal minors of $ \mathcal{H}_g^U$ is identically zero.
\end{proof}

\begin{lemma}
  \label{lem:signature_symmetric_matrix_rank}
  Let $S$ be a symmetric matrix in $\BR^{\delta \times \delta}$ of
  rank $r$ and let $S_i$ be its $i$-th leading principal minor for
  $0 \le i \le \delta$. We assume that $S_i \neq 0$ for $i \le r$.
  Then, the signature of $S$ equals $r - 2v$ where $v$ is the number
  of sign variations in $S_0, \dots, S_r$.
\end{lemma}

\begin{proof}
  We have
  \begin{align*}
    S &=
        \begin{bmatrix}
          \widetilde{S} & V^t \\ V & W
        \end{bmatrix}
        =
        \begin{bmatrix}
          I_r & 0 \\ V \widetilde{S}^{-1} & I_{\delta - r}
        \end{bmatrix}
        \underbrace{ \begin{bmatrix}
          \widetilde{S} & 0 \\ 0 & W - V \widetilde{S}^{-1}V^t
        \end{bmatrix}}_{N}
        \begin{bmatrix}
          I_r & \widetilde{S}^{-1}V^t \\ 0 & I_{\delta - r}.
        \end{bmatrix}.
  \end{align*}
  Thus $S$ and $N$ have same signature and rank. Since $\det(\widetilde{S}) = S_r
  \neq 0$, we have $\rk(\widetilde{S}) = r = \rk(S)$. Therefore, $W -
  V \widetilde{S}^{-1}V^t = 0$ and $\Signature(S)=  \Signature(\widetilde{S})$.
  By Theorem~\ref{thm:signature_symmetric_matrix}, $\Signature(\widetilde{S}) = r -
  2v$. %%This completes the proof.
\end{proof}

We use the previous lemmas as follows. Assume that after picking randomly a
matrix $U \in \GL_\delta(\BC)$, the first $r$ leading principal minors of
$\mathcal{H}_g^U \coloneqq U^t \cdot \mathcal{H}_g \cdot U $ are not identically
zero, with $r$ the rank of $\mathcal{H}_g$. Then over a connected component of
the semi-algebraic set defined by the complementary of $\mathcal{W}_\infty$ and
the non-vanishing set of these minors, the sign of each leading principal minor
is invariant. Consequently the Tarski-query $\TQ(\mbf{f}(\eta,\cdot),
g(\eta,\cdot))$ is invariant when $\eta$ ranges over this connected component.
Then by sampling at least one point in each connected component using the
algorithm in~\cite{Le2022RealRootClassification} originating from
\cite{SafeySchost2003OnePointPerConnectedComponent}, we are able to recover
all the sign conditions satisfied by a family of polynomials $\mbf{g}$ on a
dense subset of $\mathcal{V}_\BR$ the real algebraic set defined by $\mbf{f} =
0$.

% \begin{proposition}
  % \label{prop:invariance_tarski_query_connected_component}
  % Let $r$ be the rank of $\mathcal{H}_g$ and $U \in \GL_\delta(\BQ)$
  % such that none of the $r$ first leading principal minors of
  % $\mathcal{H}_g^U$ is identically zero.  Let $M_i$ be the numerator
  % of the $i$-th leading principal minor of $\mathcal{H}_g^U$.  Let
  % $\mathcal{W}$ be the following semi-algebraic set
  % \begin{equation}
    % \label{eq:W_sa_set_non_vanishing_minors}
    % \mathcal{ W} \coloneqq \BR^t \setminus \left( \bigcup_{i=1}^r V(M_i) \cup \mathcal{W}_\infty\right),
  % \end{equation}
  % where $V(M_i)$ denotes the vanishing set of $M_i$.  Then for any
  % connected component $\mathcal{C}$ of $\mathcal{W}$, the Tarski-query
  % $\TQ(\mbf{f}(\eta,\cdot), g(\eta, \cdot))$ is invariant when $\eta$
  % ranges over $\mathcal{C}$ and equals $r - 2v$ with $v$ the number of
  % sign variations in the sequence of leading principal minors of
  % $\mathcal{H}_g^U(\eta)$.
% \end{proposition}

% \begin{proof}
%   Let $H_i$ denote the $i$-th leading principal minor of
%   $\mathcal{H}_g^U$.  Over a connectoed component $\mathcal{C}$ of
%   $\mathcal{W}$, the signs of $H_i(\eta)$ are constant. So the number
%   of sign variation $v$ in $(1, H_1(\eta), \dots, H_r(\eta))$ is also
%   constant.  Next, the Tarski-query
%   $\TQ(\mbf{f}(\eta,\cdot), g(\eta, \cdot)) =
%   \Signature(\mathcal{H}_g(\eta))$ by
%   \cref{cor:specialized_hermite_matrix_tarski_queries}.  But
%   $ \Signature(\mathcal{H}_g(\eta)) =
%   \Signature(\mathcal{H}_g^U(\eta)) = r -2v$ by
%   \cref{lem:signature_symmetric_matrix_rank}. And the result follows.
% \end{proof}

Algorithm~\ref{alg:classification} for solving Problem~\ref{pb:solution_classification}
uses the subroutines:\\
$\bullet$ $\mathbf{FirstHermiteMatrix}$ follows from Algorithm~1 in
\cite{Le2022RealRootClassification}. It takes as input a polynomial sequence
$\mbf{f}$ that satisfies Assumption~\ref{assump:finite_nb_sol_generically} and outputs a
Gr{\"o}bner basis $\GB$ of $\left< \mbf{f} \right> \subseteq \BQ[\mbf y,\mbf x]$ for the ordering
$\texttt{grevlex}(\mbf{x}) \succ \texttt{grevlex}(\mbf{y})$,
a monomial basis $\Basis$
of $\mathcal{A}_\BK$ derived from $\GB$,
the family of multiplication matrices
$(M_b)_{b \in \Basis}$ in $\Basis$,
a polynomial $w_\infty \in \BQ[\mbf{y}]$ whose
vanishing set is $\mathcal{W}_\infty$ defined in (\ref{eq:W_infinity}), % (for
% example $w_\infty$ can be defined as the square-free part of the least common
% multiple of $\lc_{\mbf{x}}(p)$ for $p \in G$)
and the Hermite matrix $\mathcal{H}_1$
associated to $(\mbf{f},1)$ w.r.t. $\Basis$.\\
$\bullet$ $\mathbf{LeadPrincMinors}$ returns the list of the numerators
of the nonzero leading principal minors of a matrix with entries
in $\BQ(\mbf{y})$.\\
$\bullet$ $\mathbf{SamplePoints}$ that takes as input a sequence of polynomials
$(h_1, \dots, h_\ell)\subset \BQ[\mbf{y}]$ and sample a finite set of points
that meets every connected component of the semi-algebraic set defined
by $h_1 \neq 0, \dots, h_\ell \neq 0$.\\
For a family of polynomials $(Q_1, \dots, Q_\ell) \subset \BQ[\mbf{y}]$ and
$\eta \in \BR^t$ the sign pattern of $(Q_i(\eta))_{1 \le i \le \ell}$ is the
semi-algebraic formula $\Phi$ below
  \begin{equation}
    \label{eq:sign_pattern}
    \Phi \coloneqq \bigwedge_{i=1}^\ell \sign(Q_i) = \sign(Q_i(\eta)). 
  \end{equation}

\begin{algorithm}[t]
  \caption{Classification}
  \label{alg:classification}
  \Input{ \setlist{nolistsep}
    \begin{itemize}[noitemsep,label=$-$]
    \item A polynomial sequence
      $\mbf{f} =(f_1, \dots, f_e) \subset \BQ[\mbf{y}][\mbf{x}]$ such
      that $\mbf{f}$ satisfies Assumption~\ref{assump:finite_nb_sol_generically} \\
    \item A polynomial sequence  $\mbf{g} =(g_1, \dots, g_s) \subset \BQ[\mbf{y}][\mbf{x}]$\\
    \end{itemize}
  } \Output{ \setlist{nolistsep}
    \begin{itemize}[noitemsep,label=$-$]
    \item A set $\Sigma \subseteq \{0,1,-1\}^s$ of sign conditions
      satisfied by $\mbf{g}$ on the real algebraic set defined by
      $\mbf{f}$
    \item The description of a collection of semi-algebraic sets
      $\mathcal{T}_i$ solving Problem~\ref{pb:solution_classification}
    \end{itemize}
  }
  $\mathcal{H}_1, w_\infty, \GB, \Basis, (M_b)_{b \in \Basis}
  \leftarrow \mathbf{FirstHermiteMatrix}(\mbf{f})$
  \label{line:first_hermite_matrix}\\
  Choose randomly a matrix $U \in \BQ^{\delta \times \delta}$ \\
  $\Sigma \leftarrow \emptyset ,~ \Ada \leftarrow \emptyset, ~ M \leftarrow \Mat(\Ada,\Sigma)$ \\
  Minors $\leftarrow \emptyset$ \\
  \For{$i$ \textbf{from} $1$ \textbf{to} $s$ \label{line:for_main_loop}}{
    $ \mbf{g}_i \leftarrow (g_{s-i+1}, \dots, g_s)$ \label{line:def_gi}\\
    $ \Sigma \leftarrow \{0,1,-1\} \times \Sigma, ~ \Ada \leftarrow  \{0,1,2\} \times  \Ada, $ \\
    $ M \leftarrow \Mat(\Ada,\Sigma) = \Mat( \{0,1,2\} ,\{0,1,-1\} ) \otimes M$    \\
    \For{$\alpha \in \Ada$\label{line:for:Ada}}{ Compute $\mathcal{H}_{\mbf{g}_i^{\alpha}}$ using
      the algorithm of Section~\ref{ssec:basicparam} % for computing parametric Hermite matrices
      and let $\rk_\alpha$ be its rank\label{line:rank} \\
      $(h_1^\alpha, \dots, h_{\rk_\alpha}^\alpha) \leftarrow
      \mathbf{LeadPrincMinors} (U^t \cdot
      \mathcal{H}_{\mbf{g}_i^{\alpha}} \cdot U)$  \label{line:minor_hermite_matrix} \\
      Minors $\leftarrow$ Minors $\cup~ \{h_1^{\alpha}, \dots,
      h_{\rk_\alpha}^\alpha\}$ } $L \leftarrow \mathbf{SamplePoints} (w_\infty
    \neq 0
    ~\land %\bigwedge_{\alpha \in \Ada} \bigwedge_{j=1}^{\rk_\alpha} h_j^\alpha \neq 0 )
    $ Minors $\neq 0)$ \label{line:sample_points}
    \\
    \For{$\eta \in L$}{
      $T_\eta \leftarrow \left( \Signature(\mathcal{H}_{\mbf{g}_i^\alpha}(\eta)) \right)_{\alpha \in \Ada}$\\
      Solve $M \cdot c_\eta = T_\eta$ to compute $c_\eta$ \\
      Deduce $\Sigma_\eta$ corresponding to nonzero entries in $c_\eta$ }
    $\Sigma \leftarrow \bigcup_{\eta \in L} \Sigma_\eta$ \\
    Delete in $M$ columns whose index is not in $\Sigma$ \\
    Deduce $\Ada = \Ada(\Sigma)$ from the row rank profile of $M$ and delete the
    other rows } \For{$\eta \in L$}{ $\Phi_\eta \leftarrow$ Sign pattern of
    $(h_j^{\alpha}(\eta))$ for $ \alpha \in \Ada$ and $1 \le j \le \rk_\alpha$
    as in (\ref{eq:sign_pattern})\\
    $r_\eta \leftarrow$ entry of $c_\eta$ corresponding to $(1,1,\dots, 1) \in \{0,1,-1\}^s$ \\
  } \Return{$ \Sigma, ~(\Phi_\eta \land w_\infty \neq 0 \land \mathrm{Minors}
    \neq 0, \eta, r_\eta)_{\eta \in L}$ }
\end{algorithm}

% \todo{Proof of correctness}

\begin{theorem}[Correctness]
  \label{thm:correctness_classification_algo}
  Assume that $\mbf{f}$ satisfies Assumption~\ref{assump:finite_nb_sol_generically}. Let
  $\mbf{g} = (g_1, \dots, g_s)$ be a polynomial sequence. There exists a Zariski
  dense subset $\mathcal{U}$ of $\GL_\delta(\BC)$ such that if $U$ is sampled in
  $\mathcal{U} \cap \BQ^{\delta \times \delta}$, Algorithm~\ref{alg:classification}
  outputs a set $\Sigma$ that is equal to $\SIGN(\mbf{g}, \mathcal{V}_\BR \cap
  \pi^{-1}(\mathcal{W}))$ where $\mathcal{V}_\BR$ is the real algebraic set
  defined by $\mbf{f}$ and $\mathcal{W}$ is a nonempty Zariski open subset of
  $\BC^{t}$; and a solution to Problem~\ref{pb:solution_classification} for $(\mbf{f},
  \mbf{g})$.
\end{theorem}

\begin{proof}
  For $ 1 \le i \le s$, let $\Sigma_i$ be the value of $\Sigma$ and $\Ada_i$ the
  value of $\Ada$ after the $i$-th iteration of the loop in
  line~\ref{line:for_main_loop}. Let $g_i$ be as defined in line~\ref{line:def_gi}.
  We also define $\Sigma_{0} \coloneqq \emptyset$
  and $\mbf{g}_{0} \coloneqq \emptyset$.

  We prove the following loop invariant: for all $0 \le i \le s$,
  there exists a Zariski dense subset $\mathcal{U}_i$ of
  $\GL_\delta(\BC)$ s.t. if $U$ was sampled in
  $\mathcal{U}_i \cap \BQ^{\delta \times \delta}$, then
  $\Sigma_i = \SIGN(\mbf{g}_i , \mathcal{V}_\BR \cap
  \pi^{-1}(\mathcal{W}_i))$ where $\mathcal{W}_i$ is the nonempty
  Zariski open subset in $\BC^t$ defined as the non-vanishing locus of
  $w_\infty$ and the polynomials in the set Minors obtained after performing the
  for loop starting at line~\ref{line:for:Ada}. % at
  % line~\ref{line:sample_points}
    %defined as
  % the non vanishing locus of $w_\infty$ and the polynomials
  % $(h_j^\alpha)$ defined at line~\ref{line:minor_hermite_matrix} for
  % $i \le s$ and $\mathcal{W}_{s+1} \coloneqq \BC^t$.

  This is true when entering the loop as
  $\Sigma_{0} = \emptyset = \SIGN(\emptyset, \mathcal{V}_\BR)$.

  Now suppose that the result holds for $0 \le i-1 < s$. Let $ \Sigma^* \coloneqq
  \{0,1,-1\} \times \Sigma_{i-1}$ and $\Ada^* \coloneqq \{0,1,2\} \times \Ada_{i-1}$.
  By Lemma~\ref{lem:generically_nonzero_leading_principal_minors}, for each $\alpha \in
  \Ada^*$ there exists a Zariski dense subset $\mathcal{U}_\alpha$ of
  $\GL_\delta(\BC)$ such that if $U \in \mathcal{U}_\alpha$, the first
  $\rk_\alpha$ (see line~\ref{line:rank}) leading principal minors of $U^t \cdot
  \mathcal{H}_{\mbf{g}_i^\alpha} \cdot U$ are not identically 0. Let
  $\mathcal{U}_i \coloneqq \mathcal{U}_{i-1} \cap \bigcap_{\alpha \in \Ada^{*} }
  \mathcal{U}_\alpha$. It is a Zariski dense subset of $\GL_\delta(\BC)$ and we
  further suppose that $U$ was sampled in $\mathcal{U}_i \cap \BQ^{\delta \times
    \delta}$. In particular $U \in \mathcal{U}_{i-1}$, so by the induction
  hypothesis, $\Sigma_{i-1} = \SIGN(\mbf{g}_{i-1}, \mathcal{V}_\BR \cap
  \pi^{-1}(\mathcal{W}_{i-1}))$.
  % We let $\mathcal{W}_i$ be the Zariski open
  % subset of $\BC^t$ defined by the non vanishing locus of $w_\infty$ and the
  % polynomials in Minors at line~\ref{line:sample_points}.
  %
  Note that
  $\mathcal{W}_i$ is nonempty since all its defining polynomials are not
  identically 0 and $\mathcal{W}_i \subseteq \mathcal{W}_{i-1} $ as the set
  Minors can only increase along the iterations. We now show that $\Sigma_i =
  \SIGN(\mbf{g}_i, \mathcal{V}_\BR \cap \pi^{-1}(\mathcal{W}_i)) $. Let
  $\mathcal{R}$ be the semi-algebraic set defined as the real trace of
  $\mathcal{W}_i$.
  % For all $\alpha \in \Ada^*$, let $\rk_\alpha$ be the rank of
  % $\mathcal{H}_{\mbf{g}_i^\alpha}$ (line \ref{line:rank}).
  Then the signs of the
  first $\rk_\alpha$ leading principal minors of $U^t
  \cdot\mathcal{H}_{\mbf{g}_i^\alpha}(\eta) \cdot U$ are invariant when $\eta$
  ranges over a connected component $\mathcal{C}$ of $\mathcal{R}$. By
  Lemma~\ref{lem:signature_symmetric_matrix_rank}, the vector $T_\eta \coloneqq
  (\Signature(\mathcal{H}_{\mbf{g}_i^\alpha}(\eta)))_{\alpha \in \Ada^*}$ is
  invariant when $\eta$ varies over $\mathcal{C}$. However, the set $L$ defined
  at line~\ref{line:sample_points} contains at least one point in each connected
  component of $\mathcal{R}$. So we have $\{T_\eta \mid \eta \in \mathcal{R} \}
  =\{T_\eta \mid \eta \in L \} $. Moreover, by the induction hypothesis, for all
  $\eta \in \mathcal{R}$, we have $\SIGN(\mbf{g}_{i-1},\mathcal{V}_\BR \cap
  \pi^{-1}(\eta)) \subseteq \Sigma_{i-1} $ since $\eta \in \mathcal{W}_{i-1}$.
  Thus for all $\eta \in \mathcal{R}$, we have
  \begin{align*}
    \Sigma_\eta \coloneqq \SIGN( \mbf{g}_{i},\mathcal{V}_\BR \cap \pi^{-1}(\eta)) 
    \subseteq \{0,1,-1\} \times \Sigma_{i-1} \eqqcolon \Sigma^*.
  \end{align*}
  As a consequence of Proposition~\ref{prop:sign_determination_general_system},
  $\Sigma_\eta$ corresponds to the nonzero entries of $c_\eta \coloneqq
  \Mat(\Sigma^*,\Ada^*)^{-1} \cdot T_\eta = M^{-1} \cdot T_\eta$. Since $M$ does not
  depend on $\eta$, it holds that $\Sigma_\eta$ is invariant over a connected
  component of $\mathcal{R}$. Finally,
  \begin{align*}
    \SIGN(\mbf{g}_i, \mathcal{V}_\BR \cap \pi^{-1}(\mathcal{W}_i))
    &= \SIGN(\mbf{g}_i, \mathcal{V}_\BR \cap \pi^{-1}(\mathcal{R})) \\
    &= \bigcup_{\eta \in \mathcal{R}}  \Sigma_\eta
    =  \bigcup_{\eta \in L}  \Sigma_\eta = \Sigma_i.
  \end{align*}
  Hence Algorithm~\ref{alg:classification} outputs a set $\Sigma$ describing
  all the sign conditions realised by $\mbf{g}$ on $\mathcal{V}_\BR \cap
  \pi^{-1}(\mathcal{W})$ for $\mathcal{W}$ a nonempty Zariski open subset of
  $\BC^t$. Finally, we show that the output of Algorithm~\ref{alg:classification}:
  $(\Phi_\eta \land w_\infty \neq 0 \land \mathrm{Minors} \neq 0, \eta,
  r_\eta)_{\eta \in L}$ is a solution for Problem~\ref{pb:solution_classification}. For
  $\eta \in L$, let $\mathcal{T}_\eta$ be the semi-algebraic set defined by
  $\Phi_\eta \land w_\infty \neq 0 \land \mathrm{Minors} \neq 0$. By
  construction, $c_{\eta'}$ is invariant when $\eta'$ varies over
  $\mathcal{T}_\eta$ and its first entry equals $r_{\eta}$ that is exactly
  $\sharp \mathcal{V}_\BR \cap \pi^{-1}(\eta')$.
  % that satisfy $g_1 > 0,  \dots, g_s > 0$.
  The union of the sets $\mathcal{T}_\eta$ is $\BR^t \setminus \mathcal{W}$; it
  is dense in $\BR^t$.
\end{proof}

\subsection{Complexity analysis}

Further, we use the following notation for integers $a, b, c$:
\[\mathscr{T}_{a, b, c} = \binom{a+b+c}{a}\text{ and
  }\mathscr{M}_{a, b}=\binom{a+b}{a}.\]
{Note that one can evaluate a multivariate polynomial of degree at most $D$ in $k$
  variables within $\mathcal{O}(\mathscr{M}_{D, k})$ arithmetic operations.} 

Let $\mbf{f} = (f_1 ,\dots, f_e) \subseteq \BQ[\mbf{y}][\mbf{x}]$ be a regular
sequence satisfying
Assumptions~\ref{assump:finite_nb_sol_generically} and~\ref{assump:generic_degree_grobner_basis} and
$\mbf{g} = (g_1 ,\dots, g_s) \subseteq \BQ[\mbf{y}][\mbf{x}]$ a polynomial
sequence. Let $\dd$ be a bound on the degree of the polynomials in $\mbf{f}$ and
$\mbf{g}$. We further assume that $n,t$ and $\dd$ are
at least $2$ as we are dealing with asymptotics. Let $2 < \omega \le 3$ be
an admissible exponent for matrix multiplication. We also denote $\lambda
\coloneqq n(d-1)$.

We prove that the arithmetic cost of each loop iteration in Algorithm~\ref{alg:classification}
is dominated by the computation of
the sample points at line~\ref{line:sample_points}.
By \cite[Prop.~26]{Le2022RealRootClassification}, the cost in terms of
arithmetic operations in $\BQ$ of the call to $\mathbf{FirstHermitematrix}$
at line~\ref{line:first_hermite_matrix} is at most
% \begin{align}
%   \label{eq:cost_first_hermite_matrix}
%   \softO \left(  \binom{t + 2 \lambda}{t}
%   \left(  n \binom{n + t  + d }{n + t} + n^{\omega + 1}d^{\omega n +1 } + d^{(\omega + 1 ) n }\right) \right). 
% \end{align}
\begin{align}
  \label{eq:cost_first_hermite_matrix}
  \softO \left( \mathscr{M}_{t, 2\lambda} 
  \left(  n \mathscr{T}_{d,t,n} + n^{\omega + 1}d^{\omega n +1 } + d^{(\omega + 1 ) n }\right) \right). 
\end{align}
Note that at loop iteration $i$ the newly computed Hermite matrices are of the
form $\mathcal{H} \cdot L_{g_i}, \mathcal{H}\cdot L_{g_i}^2$, where
$\mathcal{H}$ is a Hermite matrix that has already been computed at the previous
iteration and $L_{g_i}$ is the matrix of the multiplication by $g_i$ w.r.t. the
basis $\Basis$ in $\AK$. So, each Hermite matrix is computed by multiplying a known
Hermite matrix by one matrix of multiplication $L_{g_i}$ for $1 \le i \le s$.
%We start by bounding the cost of computing each matrix of multiplication we need.

\begin{lemma}
  \label{lem:cost_multiplication_gi}
  % \todo{Bound the cost of computing $L_{g_i}$.}
  Under the above assumptions,
  let $g$ be one of the $g_i$'s, one can compute the
  matrix of multiplication $L_g$ w.r.t. basis $\Basis$ within
  % \begin{equation*}
  %   \softO\left(  \binom{t + d + \lambda}{t}\left(\binom{ n+ t + d}{n+t} +
  %       \binom{n + d}{d} d^{n \omega}+ n^{\omega  +1 } d^{n\omega +1}\right)\right)
  % \end{equation*}
  \begin{equation*}
    \softO\left( \mathscr{T}_{t,d,\lambda} \left(\mathscr{T}_{d,t,n} +
        \mathscr{M}_{n,d} d^{ \omega n}+ n^{\omega   +1 } d^{ \omega n +1}\right)\right)
  \end{equation*}
  arithmetic operations in $\BQ$.
\end{lemma}

\begin{proof}
  Let $\delta$ be the size of the Hermite matrix $L_g$. We
  already observed that $\delta \le \dd^n$. We compute $L_g$ by
  evaluation and interpolation using the multivariate interpolation algorithm
  of~\cite{canny1989solving}. Because $\mbf{f}$ satisfies
  Assumption~\ref{assump:generic_degree_grobner_basis} and is regular, by
  Lemma~\ref{lem:degree_bound_multiplication_matrix} and
  \cite[Lem.~23]{Le2022RealRootClassification}, the matrix $L_g$ has polynomial
  entries in $\mbf{y}$ of degree at most $d + \lambda$. Thus we need
  $\mathscr{T}_{t,d,\lambda}$ % $\binom{t + d + \lambda}{t}$
  interpolation points $\eta \in \BQ^t$.

  First we bound the cost of computing $L_{g}(\eta)$ for $\eta \in \BQ^t$. We
  start by computing all the matrices $L_{x_i}(\eta)$.
  This is done in time $\mathcal{O}(dn^{\omega + 1} \delta^\omega) =
  \mathcal{O}(n^{\omega + 1} d^{\omega n + 1}) $
  using~\cite[Algo.~4]{FGHR14}. Then we evaluate $g$ at $\eta$ in
  time $\mathcal{O}\left( \mathscr{T}_{d,t,n}% \binom{n + t
                                                                  % +d}{n+t}
  \right) $. We write $g(\eta, \mbf{x}) = \sum_{m} c_m m$ where $c_m \in \BQ$
  and $m$ ranges over the set of monomials in $\mbf{x}$ of degree at most $d$.
  There are $\mathscr{M}_{n,d}$ % $\binom{n + d}{d} $
  such monomials. We compute all the matrices $L_m(\eta)$ using
  $O(\mathscr{M}_{n,d}% \binom{n + d}{d}
  d^{\omega n })$ arithmetic operations by multiplying appropriately the matrices
  $L_{x_i}(\eta)$. Then, we compute $L_{g}(\eta) = \sum_m c_m L_m(\eta)$ in time
  $\mathcal{O}(d^{2n}\mathscr{M}_{n,d}% \binom{n + d}{d}
  )$. %\todo{Check}
  All in all, computing $L_{g}(\eta)$ uses $\mathcal{O}\left(
    \mathscr{T}_{d,t,n}% \binom{ n+ t + d}{n+t}
    + \mathscr{M}_{n,d}% \binom{n + d}{d}
    d^{ \omega n}+ n^{\omega +1 } d^{\omega  n +1}\right)$ arithmetic operations in
  $\BQ$. Hence, the whole evaluation step has an arithmetic cost lying in
  \begin{equation*}
    \mathcal{O}\left(  \mathscr{T}_{t,d,\lambda}% \binom{t + d + \lambda}{t}
      \left(\mathscr{T}_{d,t,n}%\binom{ n+ t + d}{n+t}
        + \mathscr{M}_{n,d}% \binom{n + d}{d}
        d^{ \omega n }+ n^{\omega  +1 } d^{\omega n +1}\right)\right).
  \end{equation*}
  Finally, we interpolate $\delta^2$ entries which are polynomials in
  $\BQ[\mbf{y}]$ of degree at most $d + \lambda$. So using multivariate
  interpolation~\cite{canny1989solving},
  this is done in time $\mathcal{O}\left( \delta^2
    \mathscr{T}_{t,d,\lambda} %\binom{t + d + \lambda}{t }
    \log^2 \mathscr{T}_{t,d,\lambda}% \binom{t + d + \lambda}{t }
    \log \log  \mathscr{T}_{t,d,\lambda}% \binom{t + d + \lambda}{t }
  \right)$. Summing the cost of the two
  steps together ends the proof.
\end{proof}

\begin{proposition}
  \label{prop:cost_hermite_matrix_computation}
  Suppose that $\mbf{f} = (f_1, \dots, f_e ) \subset \BQ[\mbf{y}][\mbf{x}]$ is a
  regular sequence satisfying
  Assumptions~\ref{assump:finite_nb_sol_generically} and~\ref{assump:generic_degree_grobner_basis}.
  Then any parametric Hermite matrix $\mathcal{H}_g$ occurring in
  Algorithm~\ref{alg:classification} with $g \in \BQ[\mbf{y}][\mbf{x}]$ of degree $\dd_g$
  can be computed within
  \begin{align}
    \label{eq:cost_Hg}
    \softO \left(\mathscr{T}_{t, \dd_g, 2\lambda}% \binom{t + d_g + 2 \lambda}{t}
    \left( d^{2n} \mathscr{T}_{t, \dd_g, 2\lambda} %\binom{t + d_g + 2 \lambda}{t}
    + \dd^{\omega n} \right) \right)
  \end{align}
  arithmetic operations in $\BQ$.
  Moreover, any minor of $\mathcal{H}_g$ can be computed using
  \begin{align}
    \label{eq:cost_minors}
    \softO \left( \mathscr{M}_{t, (\dd_g +  \lambda)d^n}% \binom{t + (d_g +  \lambda)d^n}{t}
    \left( \dd^{2n}
    \mathscr{T}_{t, \dd_g, 2\lambda}%\binom{t + d_g + 2 \lambda}{t}
    + \dd^{\omega n} \right) \right)
  \end{align}
  arithmetic operations in $\BQ$.
  %In particular the cost of computing the matrix $\mathcal{H}_g$ is dominated
  %by the cost of computing its minors.
\end{proposition}

\begin{proof}
  Again let $\delta \le \dd^n$ be the size of the Hermite matrix
  $\mathcal{H}_g$. We write $g = g' g_i$ so that $\mathcal{H}_g =
  \mathcal{H}_{g'} \cdot L_{g_i}$ for some $ 1 \le i \le s$ and
  $\mathcal{H}_{g'}$ a parametric Hermite matrix that is already known. By
  Lemma~\ref{lem:degree_bound_multiplication_matrix} and
  Proposition~\ref{prop:degree_bound_hermite_matrix}, the matrices $\mathcal{H}_{g},
  \mathcal{H}_{g'}$ and $L_{g_i}$ have entries in $\BQ[\mbf{y}]$. Moreover, by
  \cite[Lem.~23]{Le2022RealRootClassification} the largest degree among the
  entries of $\mathcal{H}_g$ and $\mathcal{H}_{g'}$ is bounded by $\Lambda
  \coloneqq d_g + 2 \lambda$ and $L_{g_i}$ has all its entries of degree at most
  $d + \lambda$. We compute the evaluations $\mathcal{H}_g(\eta) =
  \mathcal{H}_{g'}(\eta) \cdot L_{g_i}(\eta)$ for $\mathscr{M}_{t,
    \Lambda}% \binom{t + \Lambda}{t}
  $ distinct points $\eta \in \BQ^t$, and then we interpolate the matrix
  $\mathcal{H}_{g}$ using the algorithm of~\cite{canny1989solving}.

  Let $\eta \in \BQ^t$. We first estimate the cost of computing $L_{g_i}(\eta)$.
  %We start by computing all the matrices representing the $L_{x_i}(\eta)$.
  %This can be done in $\mathcal{O}(dn^{\omega + 1} \delta^\omega) = \mathcal{O}(n^{\omega + 1} d^{\omega n + 1}) $
  %using~\cite[Algo.~4]{faugere2013polynomial}. Then we write $g_{i}(\eta, \mbf{x}) = \sum_{m} c_m m$ where
  %$c_m \in \BQ$ and $m$ ranges over the set of monomials in $\mbf{x}$ of degree at most $d$.
  %There are  $\binom{n + d}{d} $ such monomials. So we compute all the matrices $L_m(\eta)$
  %with $O(\binom{n + d}{d} d^{n\omega +1})$ arithmetic operations. Then, we have
  %$L_{g_i}(\eta) = \sum_m c_m L_m(\eta)$
  %and this can be computed in $\mathcal{O}(d^{2n}\binom{n + d}{d})$  arithmetic operations.
  %So we are able to compute $L_{g_i}(\eta)$ using
  %$\mathcal{O}\left(\left(\binom{n + d}{d} + n^{\omega  +1 } \right)d^{n\omega +1}\right)$
  %arithmetic operations in $\BQ$.
  It is the evaluation of $\delta^2$ polynomials in $\BQ[\mbf{y}]$ of degree at
  most $d + \lambda$. So its cost is in $\mathcal{O} \left(
    \mathscr{T}_{t,d,\lambda}% \binom{t + d + \lambda }{t}
    \delta^2\right) $ arithmetic operations in $\BQ$.
  % \todo{We should refer to  something}
  
  Similarly, we estimate the cost for computing
  $\mathcal{H}_{g'}(\eta)$.
  We obtain $\mathcal{O} \left(
    \mathscr{M}_{t, \Lambda}% \binom{t + \Lambda}{t}
    \delta^2\right)$ arithmetic operations in $\BQ$.

  Finally, we need to compute the matrix product $\mathcal{H}_{g'}(\eta)
  L_{g_i}(\eta)$ and this is done in time $\mathcal{O}(\delta^\omega)$.
  Notice that $\lambda - d = n(d-1) - d = nd - n - d \ge 0$, as $n \ge 2$ and $ d \ge 2$. So
  $d + \lambda \le \Lambda$. Summing up every step together we obtain that the evaluation 
  $\mathcal{H}_g(\eta)$ can be computed within
  $\mathcal{O} \left( \delta^2 \mathscr{M}_{t, \Lambda}% \binom{t + \Lambda}{t }
    + \delta^\omega \right)$ arithmetic operations in $\BQ$.
  Since there are $\mathscr{M}_{t, \Lambda}$ evaluation points, the whole evaluation step uses
  \begin{equation*}
    \mathcal{O} \left( \mathscr{M}_{t, \Lambda}% \binom{t + \Lambda}{t }
      \left( d^{2n} \mathscr{M}_{t, \Lambda}% \binom{t + \Lambda}{t }
        + d^{\omega n} \right) \right)
  \end{equation*}
  arithmetic operations in $\BQ$ at most. Finally, we interpolate $\delta^2$
  entries which are polynomials in $\BQ[\mbf{y}]$ of degree at most $\Lambda$.
  Using~\cite{canny1989solving}, the complexity of this step lies in
  $\mathcal{O}\left( \delta^2 \mathscr{M}_{t, \Lambda}% \binom{t +
      % \Lambda}{t }
    \log^2 \mathscr{M}_{t, \Lambda}% \binom{t + \Lambda}{t }
    \log \log \mathscr{M}_{t, \Lambda}% \binom{t +
      % \Lambda}{t }
  \right)$. Summing up these costs, we obtain the claimed complexity for
  $\mathcal{H}_g$.

  To compute the minors, we again use an evaluation-interpolation scheme. Any
  minor of $\mathcal{H}_g$ has degree at most $(d_g + \lambda)d^n$ by
  Corollary~\ref{cor:degree_minor_hermit_matrix_regular}, so we need $\mathscr{M}_{t,
    (d_g + \lambda)d^n}
  % \binom{t + (d_g + \lambda)d^n}{t}
  $ interpolation points. Each evaluation of the matrix costs
  $\mathcal{O}\left(\delta^2 \mathscr{M}_{t, \Lambda}% \binom{t + \Lambda}{t}
  \right)$ and the computation of the minors lies in
  $\mathcal{O}(\delta^\omega)$. We deduce the bound as before.
\end{proof}

% Let $g \in \BQ[\mbf{y}][\mbf{x}] $ be a polynomial such that $\mathcal{H}_g$ occurs in \cref{alg:classification}.
% We note that the complexity of computing the matrices $\mathcal{H}_g$ and $L_{g_i}$
% is bounded by the complexity of computing the minors of $\mathcal{H}_g$.
% As in \cite[Sec. 6.2]{Le2022RealRootClassification}, we have that
% \begin{equation*}
%   \binom{d+n+t}{n+t} \le \binom{t+  \lambda d^n}{t}(2 d^n) \le  \binom{t+  (d_g + \lambda) d^n}{t}(2 d^n) 
% \end{equation*}

One shows that the cost for computing the minors of $\mathcal{H}_g$
dominates the cost for computing the matrix $\mathcal{H}_g$.
First note that
\begin{align*}
  \mathscr{M}_{n, d}% \binom{n+d}{d}
  = \frac{(d+n)\dots (d+1)}{n!} = d^n \prod_{k=1}^{n}\left( \frac{1}{d}+\frac{1}{k} \right)
  % \left(\frac{1}{d} + \frac{1}{n}\right)\left(\frac{1}{d} + \frac{1}{n-1}\right)\dots
  % \left(\frac{1}{d} + 1\right)
  \le  2 d^n,
\end{align*}
since $\frac{1}{d} + 1 \le 2$ and $\frac{1}{d} + \frac{1}{k} \le 1$ for $k \ge
2$. Then, the cost in Lemma~\ref{lem:cost_multiplication_gi} is bounded by the cost
(\ref{eq:cost_first_hermite_matrix}) of $\mathbf{FirstHermiteMatrix}$. In
\cite[Sec. 6.2]{Le2022RealRootClassification}, it is shown that
(\ref{eq:cost_first_hermite_matrix}) is bounded by $\softO \left(
  \mathscr{M}_{t, \lambda d^n}% \binom{t +
    % \lambda d^n}{t}
  \mathscr{M}_{t, 2\lambda}% \binom{t + 2 \lambda}{t }
  d^{2n} \right)$. Thus the cost for computing $L_{g_i}$ is bounded by
(\ref{eq:cost_minors}). In addition, it holds that (\ref{eq:cost_Hg}) is bounded
by (\ref{eq:cost_minors}). Hence we can conclude that computing the minors
dominates the cost of computing the matrices.

Now let us bound the cost of computing the set of sample points at
Line~\ref{line:sample_points}. In Algorithm~\ref{alg:classification}, we compute parametric
Hermite matrices $\mathcal{H}_{\mbf{g}^\alpha}$, with $\mbf{g}^\alpha = \prod
g_i^{\alpha_i}$ with $\alpha \in \{0,1,2\}^s$. So, the degree of $g$ is bounded
by $2ds$. Hence, the degree of any minor in $\mathrm{Minors}$ is bounded by
$\mathfrak{D} \coloneqq (2ds + \lambda)d^n$. Let $\rho_i$ be the size of $\Sigma$
at the end of iteration $i$ of the loop. By~\cite{basu2005betti}, we have
\begin{equation*}
  \rho_i \le \rho \coloneqq \binom{s}{t} 4^{t+1} d(2d -1)^{n+t -1}.
\end{equation*}
At iteration $i$, we compute at most $ 2 \rho $ new Hermite matrices,
so we add at most
$2 \delta \rho$ new minors in the set $\mathrm{Minors}$. Let $M_i$ be the
size of the set $\mathrm{Minors}$ after the loop iteration $i$. We have $ M_{i}
 \le 2 \delta  i \rho \le 2 \delta sr$.
 So we call the routine
$\mathbf{SamplePoints}$ with at most $2 d^n s \rho$ polynomials, because as
$\mbf{f}$ satisfies Assumption~\ref{assump:generic_degree_grobner_basis}, we can omit
$w_\infty =1$. By \cite[Thm.~2]{Le2022RealRootClassification}, the set of sample
points $L$ contains at most $(4 d^n s \rho \mathfrak{D} )^t$ points and this set
can be computed using
\begin{equation}
  \label{eq:cost_sample_points}
  \softO\left( \mathscr{M}_{t, \mathfrak{D}}% \binom{t  + \mathfrak{D}}{t}
    (2d^n s \rho)^{t+1} 2^{3t} \mathfrak{D}^{2t+1} \right) 
  %\softO\left( \binom{\mathfrak{D} + t}{t} 2^{4t} d^{n(t+1) } s^{t+1} \rho^{t+1}\right)
\end{equation}
arithmetic operations in $\BQ$.
We can now prove Theorem~\ref{thm:main}.

\noindent
\begin{proof}[Proof of Theorem~\ref{thm:main}]
  The sequence $\mbf{f}$ is regular and satisfies both
  Assumptions~\ref{assump:finite_nb_sol_generically} and~\ref{assump:generic_degree_grobner_basis}.
  At each iteration of Algorithm~\ref{alg:classification}, the call to
  $\mathbf{SamplePoints}$ has a cost bounded by~(\ref{eq:cost_sample_points}).
  We also compute at most $2 \rho$ new Hermite matrices and their $\delta \le d^n$
  leading principal minors. By Proposition~\ref{prop:cost_hermite_matrix_computation}, this
  can be done using
  \begin{equation*}
    \softO\left( d^n \rho \mathscr{M}_{t, \mathfrak{D}}% \binom{t + \mathfrak{D}}{t}
      \left( d^{2n} \mathscr{T}_{t, 2s\dd,2\lambda}% \binom{t + 2sd + 2 \lambda}{t}
    + \dd^{\omega n} \right) \right) 
\end{equation*} arithmetic operations.
Since  $\mathscr{T}_{t,
  2s\dd,2\lambda}% \binom{t + 2sd + 2 \lambda}{t}
\in \mathcal{O}\left( \mathfrak{D}^t \right)$,
the above estimate is bounded by~(\ref{eq:cost_sample_points}).
Next, we have to evaluate the signatures of at most $3 \rho$ Hermite matrices for
every points $\eta \in L$. This is done by evaluating the sign patterns of the
minors. There are at most $3 \delta \rho$ minors of degree at most $\mathfrak{D}$
to evaluate at at most $(4d^n s \rho \mathfrak{D})^t$ points. This is done
within $\mathcal{O}\left( d^n \rho (4 d^n s\rho \mathfrak{D})^t \mathscr{M}_{t,
    \mathfrak{D}}% \binom{t+ \mathfrak{D}}{t}
\right)$ arithmetic operations in $\BQ$ and this is bounded by
(\ref{eq:cost_sample_points}). The linear algebra to solve the linear systems $M
\cdot c_\eta = T_\eta$ or to compute the row rank profile of $M$ has a
negligible cost in front of the evaluations of the minors. Finally we sum the
costs for each of the $s$ iterations and substitute the values of $\lambda, \rho,
\mathfrak{D}$ to get the complexity estimate. The algorithm outputs $\sharp L
\le (4 d^n s \rho \mathfrak{D})^t$ formulas. Each formula contains $\mathcal{O}(d^n
s \rho)$ minors of degree at most $\mathfrak{D}$. This completes the proof.
%
%
  % The complexity estimate is a direct consequence of the fact that the bound for
  % the call to $\mathbf{SamplePoints}$ dominates the cost of one loop iteration.
  % Let us prove this fact. We have at most $\delta r \le d^n r $ new minors to
  % compute and by \cref{prop:cost_hermite_matrix_computation}, this can be done
  % using
  % \begin{equation*}
  %   \softO\left( d^n r \mathscr{M}_{t, \mathfrak{D}}% \binom{t + \mathfrak{D}}{t}
  %     \left( d^{2n} \mathscr{T}_{t, 2s\dd,2\lambda}% \binom{t + 2sd + 2 \lambda}{t}
  %   + \dd^{\omega n} \right) \right) 
% \end{equation*}
% arithmetic operations. Since $\mathscr{T}_{t,
%   2s\dd,2\lambda}% \binom{t + 2sd + 2 \lambda}{t}
% \in \mathcal{O}\left( \mathfrak{D}^t \right)$, we can conclude that the above
% cost is bounded by (\ref{eq:cost_sample_points}). Then, we have to evaluate the
% signatures of at most $3 r$ Hermite matrices for every points $\eta \in L$. This
% is done by evaluating the sign patterns of the minors. There are at most $\delta
% r$ minors of degree at most $\mathfrak{D}$ to evaluate at at most $(4d^n sr
% \mathfrak{D})^t$ points. This can be done within $\mathcal{O}\left( d^n r (4 d^n
%   sr \mathfrak{D})^t \mathscr{M}_{t, \mathfrak{D}}% \binom{t+ \mathfrak{D}}{t}
% \right)$ arithmetic operations in $\BQ$ and this is bounded by
% (\ref{eq:cost_sample_points}). The linear algebra to solve the linear systems $M
% \cdot c_\eta = T_\eta$ or to compute the row rank profile of $M$ has a
% negligible cost in front of the evaluations of the minors. This completes the
% proof.
\end{proof}

\section{Practical experiments}\label{sec:experiments}

We report here on the practical behaviour of our algorithm and compare it with
existing Maple packages based on other methods for solving parametric
semi-algebraic systems. In Algorithm~\ref{alg:classification}, we need to compute sample
points per connected components of the non-vanishing set of leading principal
minors of several Hermite matrices. Once we have computed these sample points,
the semi-algebraic conditions for the classification are derived from the sign
patterns of the minors on these points. However when facing practical problems,
calling the $\mathbf{SamplePoints}$ routine with this number of minors is often
the bottleneck of Algorithm~\ref{alg:classification}. If we assume that for each
inequality $g_i$ with $1 \le i \le s$, the Hermite matrix $\mathcal{H}_{g_i}$ is
nonsingular, one can get better timings in practice with the following approach:
\begin{itemize}
\item Compute a set $\{\eta_1, \dots, \eta_\ell \}$ of sample points
  in the non-vanishing set of the determinants of
  $(\mathcal{H}_1, \mathcal{H}_{g_1}, \dots, \mathcal{H}_{g_s})$.
  For $1 \le i \le \ell$,
  perform sign determination to obtain $r_i$ the number of solutions
  of the specialized system $(\mbf{f}(\eta_i, \cdot),\mbf{g}(\eta_i, \cdot) )$.
  One can show that we obtain all the possible number of solutions that the input system can admit.
\item Next in order to get semi-algebraic conditions, compute the $3^s$ Hermite matrices
$\mathcal{H}_{\mbf{g}^\alpha}$ for all $\alpha \in \{0,1,2\}^s$ and all their
leading principal minors.
>From each sign pattern $\tau$ on this family of minors, the signatures of the Hermite matrices
are determined and 
one can associate $0 \le r_\tau \le \delta$ the number of solutions of the input system.
Finally we derive a classification from the sign patterns $\tau$ such that $r_\tau \in \{r_1,  \dots , r_\ell\}$. 
%One can directly derive a classification from the sign patterns of these minors:
%for each possible sign pattern, compute the signatures of the matrices and invert the $3^s$ linear system to get
%the number of solutions~(\cref{prop:sign_determination_general_system}).
%In particular we can not determine all the possible number of solutions of the input system.
% \item It turns out that 
\end{itemize}
Notice that we get a classification with semi-algebraic formulas that contain clauses that may be infeasible.
Yet we only need one call to the $\mathbf{SamplePoints}$ routine with $s +1$ polynomials in input.
%It is important to note that for all the applications we encounter in practice
%the number $s$ of inequalities in the input system is not \emph{too large}.
%Therefore it is more efficient to consider the $3^s$ Hermite matrices
%$\mathcal{H}_{\mbf{g}^\alpha}$ for all $\alpha \in \{0,1,2\}^s$ and compute all their
%%%leading principal minors to derive a classification from the sign patterns of these minors.
%This avoids calling the $\mathrm{SamplePoints}$ routine which is very costly.
%The drawback is that we get semi-algebraic formulas that contain clauses that may be infeasible.
%
%\todo[inline]{mais on s'en sort avec sample points sur les determinants}

The timings are given in hours (h.), minutes (m.) and seconds (s.) and the computations have been performed
on a PC Intel (R) Xeon (R) Gold 6244 CPU 3.6GHz with 1.5Tb of RAM.
%\todo{One word about the computer}
In our implementation, we compute Hermite matrices using FGb package~\cite{faugere2010fgb}
for Gr{\"o}bner basis computation.
For the sample points routine, we use RAGlib~\cite{safey2017real}.
In Table~\ref{tab:exp_generic_dense_result}, we analyse the costs on dense \emph{generic} inputs,
\ie the input polynomials
$(f_1, \dots, f_n)$ and $(g_1, \dots, g_s) \subset \BQ[\mbf{y}][\mbf{x}]$
are dense and randomly chosen among polynomials of degree $d$.
%\todo[inline]{in table 2 the systems are sparse and structured?}
%In both \cref{tab:exp_generic_dense_result,tab:exp_generic_sparse_system},
We collect results for various values of $(n,t,s,d)$.
We focus on the timings for computing all the Hermite matrices (\textbf{hm}),
all their leading principal minors (\textbf{min}).
We also report in column \textbf{det} the timings for computing only the $(s+1)$ matrices
$(\mathcal{H}_1, \mathcal{H}_{g_1}, \dots, \mathcal{H}_{g_s})$ and their determinants;
and for computing sample points (column~\textbf{sp}) in the non-vanishing locus of these determinants.
%All the timings are collected in \cref{tab:exp_generic_result}.
We compare our algorithm with the Maple packages
RootFinding[Parametric]~\cite{GeJeMo10} (the column RF)
and 
RegularChains[ParametricSystemTools] \cite{YangHX01}.

In the column RF, we give the timings for the command DiscriminantVariety~(\textbf{dv}) that
computes a set of polynomials defining a discriminant variety $\mathcal{D}$ of the input system.
For generic systems, the output of DiscriminantVariety coincides with the irreducible factors of the
determinants of $(\mathcal{H}_1, \dots, \mathcal{H}_{g_s})$ and the border polynomials returned by the
command BorderPolynomial~(\textbf{bp}) contains these
polynomials.
We also collect the results for the command CellDecomposition~(\textbf{cad}) that outputs semi-algebraic formulas
by computing an open CAD for $\BR^t \setminus \mathcal{D}$.
%The command RealRootClassification of the package RegularChains[ParametricSystemTools]
%also computes sample points

%
% The weak version of our algorithm corresponds to the case where
%we do not require the computation of semi-algebraic formulas.
%In this case, it is sufficient to sample points in the non-vanishing locus
%of the determinants of the Hermite matrices corresponding to each inequality. 
%
\begin{table}
  \begin{tabular}{cccc|cccc|cc|cc}
    & & & & \multicolumn{4}{c|}{Hermite} & \multicolumn{2}{c|}{RF} & \\
    \hline
    $n$ & $t$ & $s$& $d$& \textbf{hm} & \textbf{min} & \textbf{det} &  \textbf{sp}
        &  \textbf{dv} &  \textbf{cad} &  \textbf{bp}  \\
    \hline 
    2 & 2 & 2 & 2 & 0.15 s & 0.4 s & 0.1 s & 5 s & 0.14 s & 2 s & 0.11 s \\
    2 & 2 & 3 & 2 & 0.7 s & 2 s & 0.1 s & 10 s & 0.9 s & 10 s & 1 s \\
    3 & 2 & 1 & 2 & 0.5 s & 9 s & 0.4 s &  33 s & 10 m & 11 m  & 7 m \\
    3 & 2 & 2 & 2 & 3 s & 1 m & 0.4 s & 57 s & 10 m & 13 m & 14 m  \\ 
    2 & 3 & 2 & 2 & 0.3 s & 4 s & 0.1 s & 18m & 0.7 s & >50 h & 0.2 s   \\ 
    3 & 3 & 1 & 2 & 1 s  & 4 m & 6 s & >50 h& >50 h  & >50 h & >50 h\\
    2 & 2 & 1 & 3 & 0.9 s & 30 s & 0.8 s & 3m & 52 m & 57 m & 47 s   \\
    2 & 2 & 2 & 3 & 5 s & 5 m &  1 s & 6m & 57 m & 1h 16 m & 2 m   \\
    % 2 & 3 & 3 & 3 & 0.41 s & 84 s &
  \end{tabular}
  \caption{Generic dense system}
  \label{tab:exp_generic_dense_result}
\end{table}
%\todo[inline]{Add 3312 in table}

The column \textbf{det} has to be compared with the two columns \textbf{dv} and
\textbf{bp} as they are three different approaches to compute polynomials that
defines the boundary of semi-algebraic sets over which the number of solutions
of the input system is invariant.

We observe that our method outperforms DiscrimantVariety and BorderPolynomial.
With our approach based on the minors of the Hermite matrices, we are not only
able to solve the classification problem for systems faster by several orders of
magnitude than what can be achieved with CellDecomposition~(\textbf{cad}) and
the command RealRootClassification of the RegularChains[ParametricSystemTools]
library. We can also tackle problems that were previously out of reach.

\paragraph*{Perspective-Three-Point Problem (P3P)}

We now consider a system coming from the P3P problem and apply our
algorithm to find a classification. It should be noted that
Assumption~\ref{assump:generic_degree_grobner_basis} is not satisfied on this problem. 
The problem consists in
determining the position of a camera given the relative spatial
location of 3 control points. As in~\cite{FMRSa08}, we want to compute
a classification of the real solutions of the following system:
\begin{equation}
  \label{eq:P3P_general}
  \tag{P3P}
  \begin{cases}
    1 &= A^2 + B^2 - AB u \\
    t &= B^2 + C^2 - BC v \\
    x &= A^2 + C^2 - AC w
  \end{cases} \text{ with} \quad  A > 0, B > 0, C > 0,
\end{equation}
subject to the following constraints: $x, t > 0 $, $-2 < u, v, w <2 $,
where $A,B,C$ are the \emph{variables} and $x,t,u,v,w$ are \emph{parameters}.

In~\cite{FMRSa08}, a special case of~(\ref{eq:P3P_general}) is studied where $t
= 1$. This restriction corresponds to the case where the three controls points
form an isosceles triangle. In this case, a discriminant variety
$\mathcal{D}$ for the system is computed in \cite{FMRSa08}. Sample points in the
semi-algebraic set $\BR^4 \setminus \mathcal{D}$ in order to deduce all the
possible number of solutions of (\ref{eq:P3P_general}) in the isosceles case are
computed using RAGlib but this is not sufficient to obtain semi-algebraic
conditions that prescribe the number of real solutions to the input parametric
system.

With our method, we are able to derive these semi-algebraic descriptions for
each possible number of solutions from the signs of the leading principal minors
of parametric Hermite matrices. In less than one hour, we compute all the minors
and sample points in $\BR^4 \setminus \mathcal{D}$ whence we obtain a complete
classification in the isosceles case. We observe that despite the fact that 
Assumption~\ref{assump:generic_degree_grobner_basis} is not satisfied, the degrees of the
output formulas are smaller than the bounds predicted in the generic case.

We also studied the general case.
The system~(\ref{eq:P3P_general}) has 3 variables and 5
parameters.
We compute the first Hermite matrix $\mathcal{H}_1$ and the ones corresponding to
each inequality $\mathcal{H}_A, \mathcal{H}_B, \mathcal{H}_C$ and their determinants
in a few seconds.
This gives polynomials defining a discriminant variety of the system~(\ref{eq:P3P_general}).
Already this first step was out of reach using the Maple commands DiscriminantVariety or
BorderPolynomial. Next we are able to compute the leading principal minors of all Hermite matrices
of the form $\mathcal{H}_{A^{\alpha_1}B^{\alpha_2}C^{\alpha_3}}$
with $(\alpha_1, \alpha_2, \alpha_3) \in \{0,1,2\}^3$ and get semi-algebraic conditions
for a classification.
One further step would be to sample points outside the discriminant variety to get all the
possible number of solutions of the system~(\ref{eq:P3P_general}).

%% ============================REFERENCES=============================================%%

\balance
\bibliographystyle{ACM-Reference-Format}
\bibliography{biblio}

\end{document}